\documentclass[11pt,fleqn]{article}
\usepackage{amssymb,latexsym,amsmath,amsfonts,graphicx}

    \advance\textheight by \topskip

    \numberwithin{equation}{section}

    \def\Re{{\rm Re \,}}
    \def\Im{{\rm Im \,}}
    
    \def\I{{\rm I \,}}
    \def\II{{\rm II \,}}
    \def\III{{\rm III \,}}
    \def\IV{{\rm IV \,}}
    \def\bigO{{\cal O}}
    
    \def\Res{{\rm Res}}

\newcommand{\e}{\epsilon}
\newcommand{\lb}{\lambda}

\newcommand{\ra}{\rightarrow}

    \newtheorem{theorem}{Theorem}[section]
    \newtheorem{lemma}[theorem]{Lemma}
    
    \newtheorem{proposition}[theorem]{Proposition}
    
    \newtheorem{Definition}[theorem]{Definition}
    
    \newtheorem{Remark}[theorem]{Remark}
    \newenvironment{remark}{\begin{Remark}\rm}{\end{Remark}}
    \newtheorem{Example}[theorem]{Example}
    
    \newtheorem{Assumptions}[theorem]{Assumptions}

    \newenvironment{proof}%
    {\rm \trivlist \item[\hskip \labelsep{\bf Proof. }]}%
    {\hspace*{\fill}$\Box$\endtrivlist}
    {\rm \trivlist \item[\hskip \labelsep{\bf Proof}]}%
    {\hspace*{\fill}$\Box$\endtrivlist}

\begin{document}
\title{Universality of the break-up profile for the KdV equation in the small dispersion limit using the
Riemann-Hilbert approach}
\author{T. Claeys and T. Grava}

\maketitle

\begin{abstract}
We obtain an asymptotic expansion for the solution of the Cauchy problem for the Korteweg-de Vries (KdV) equation
\[
u_t+6uu_x+\epsilon^{2}u_{xxx}=0,\quad u(x,t=0,\epsilon)=u_0(x),
\]
for $\epsilon$ small, near the point of gradient catastrophe $(x_c,t_c)$ for the solution  of the dispersionless equation $u_t+6uu_x=0$.
 The sub-leading term in this expansion is described by the smooth solution of a fourth order ODE, which is a higher order analogue to the Painlev\'e I equation. This is in accordance with a conjecture of Dubrovin, suggesting that this is a universal phenomenon for any Hamiltonian perturbation of a hyperbolic equation. Using the Deift/Zhou steepest descent method applied on the Riemann-Hilbert problem for the KdV equation, we are able to prove the asymptotic expansion rigorously in a double scaling limit.
\end{abstract}

\section{Introduction}
It is well-known that the solution of the Cauchy problem for the Hopf equation
\begin{equation}
\label{hopf}
u_t+6uu_x=0,\quad u(x,t=0)=u_0(x), \;\;x\in\mathbb{R},\;t\in\mathbb{R}^+,
\end{equation}
reaches a point of gradient catastrophe in a finite time. The solutions of the  dissipative
 and conservative regularizations of the above hyperbolic equation display a
considerably different behavior.
Equation (\ref{hopf}) admits a Hamiltonian structure
\[
u_t +\{ u(x), H_0\}\equiv u_t +\partial_x\frac{\delta H_0}{\delta u(x)} =0,\quad
\]
with Hamiltonian and Poisson bracket given by
\[
 H_0 =\int u^3\, dx, \qquad
\{ u(x) , u(y)\}=\delta'(x-y),
\]
respectively.  The conservative regularization of (\ref{hopf}) can be obtained by
perturbing the Hamiltonian $H_0$ to the form
$H=H_0+\epsilon H_1+\epsilon^2H_2+\dots$ \cite{Getzler} (see also \cite{DMS,DZ}).
All the Hamiltonian perturbations up to the order $\epsilon^4$ of the hyperbolic equation
(\ref{hopf}) have been classified in \cite{dubcr, Lorenzoni} and
the   Hamiltonian is equal to
\[
H_{\epsilon}=\int \left[ u^3 - \epsilon^2 \frac{c(u)}{24} u_x^2
+\epsilon^4  p(u) u_{xx}^2 \right]\, dx,
\]
where  $c(u)$, $p(u)$  are two arbitrary functions of one variable.
The  equation $u_t +\{ u(x), H_{\epsilon}\}=0$ takes the form
\begin{equation}
\begin{split}
\label{riem2}
&
u_t +6u\, u_x + \frac{\epsilon^2}{24} \left[ 2 c\, u_{xxx} + 4 c' u_x u_{xx}
+ c'' u_x^3\right]+\epsilon^4 \left[ 2 p\, u_{xxxxx} \right.\\
&
\\
&\left.
+2 p'( 5 u_{xx} u_{xxx} + 3 u_x u_{xxxx}) + p''( 7 u_x u_{xx}^2 + 6 u_x^2 u_{xxx} ) +2 p''' u_x^3 u_{xx}\right]=0,
\end{split}
\end{equation}
where the prime denotes the derivative with respect to $u$.
For $c(u)=12$, $p(u)=0$, one obtains the Korteweg - de Vries (KdV)
\begin{equation}
\label{KdV}
u_t+6uu_x+\epsilon^{2}u_{xxx}=0,
\end{equation} and for $c(u)=48u $ and
$p(u)=2u$, one has the Camassa-Holm equation \cite{CH} up to order $\epsilon^{4}$.
For generic  choices of the functions $c(u)$, $p(u)$ equation (\ref{riem2})
is apparently not an integrable PDE, but however it
admits an infinite family of commuting Hamiltonians up to order $\bigO(\e^6).$

The case of small dissipative perturbations
of one-component hyperbolic equations has been well studied and understood
(see \cite{bressan} and references therein), while the  behavior of solutions
to conservative perturbations (\ref{riem2}) to the  best of our knowledge
has not been investigated after the point of gradient catastrophe of the unperturbed
equation except for the  KdV case.
The solution of the Cauchy problem for KdV in the  limit $\epsilon\ra 0$, has been studied in
the works of Gurevich and Pitaevskii, \cite{GP},
 Lax and Levermore \cite{LL}, Venakides \cite{V2}, and Deift, Venakides and Zhou \cite{DVZ, DVZ2}. The asymptotic description of
\cite{LL, DVZ} gives in general a good approximation of the KdV
solution, but it is less satisfactory near the point of gradient
catastrophe for the Hopf equation (\ref{hopf}).
Before this break-up time, solutions to the KdV  equation are,
in the small dispersion limit $\epsilon\to 0$, well approximated by solutions to the Hopf equation. After the time of gradient catastrophe, solutions to the Hopf equation cease to be well-defined for all $x$, while for any $\epsilon>0$,
  solutions to the KdV equation remain  well-defined for all $x$ and $t$
and are characterized by an oscillatory region after the gradient catastrophe. In this oscillatory region, small dispersion asymptotics for KdV solutions turn out to be elliptic \cite{V2}.
The transition from the  asymptotic regime described by (\ref{hopf}) to the
elliptic regime has not been rigorously described yet in the literature. For numerical comparisons we refer to \cite{GK,GK1}.

\medskip

This problem has been addressed by Dubrovin in \cite{dubcr}, where he formulated
the universality conjecture about the behavior of a generic solution
to any Hamiltonian perturbation of a hyperbolic equation (\ref{hopf}) near the
point $(x_c,t_c,u_c=u(x_c,t_c,0))$ of gradient catastrophe for (\ref{hopf}).
Dubrovin argued that, up to shifts, Galilean transformations and rescalings,
the behavior of the solution of (\ref{riem2}) as $\epsilon\rightarrow 0$
near the point of gradient catastrophe for the Hopf equation (\ref{hopf})  essentially depends
neither on the choice of the solution nor on the choice of the
equation. Moreover, the solution near the point $(x_c, t_c, u_c)$
is given by
\begin{equation}
\label{univer0} u(x,t,\e)\simeq u_c +\left(\dfrac{\epsilon^2c_0}{k^2}\right)^{1/7} U \left(\dfrac{
 x- x_c-6u_c (t-t_c)}{(kc_0^3\epsilon^6)^{1/7}},
\dfrac{6(t-t_c)}{(k^3c_0^2\epsilon^4)^{1/7}}\right) +O\left( \epsilon^{4/7}\right),
\end{equation}
where   $c_0$ is a constant that depends on the equation, while $k$ depends on the initial data,  and $U=U(X,T)$ is the unique real smooth solution to the fourth order ODE
\begin{equation}\label{PI20}
X=T\, U -\left[ \dfrac{1}{6}U^3  +\dfrac{1}{24}( U_{X}^2 + 2 U\, U_{XX} )
+\frac1{240} U_{XXXX}\right],
\end{equation}
which is  the second member of the Painlev\'e I hierarchy. In what follows, we will
call this equation $P_I^2$. The relevant solution is
characterized by the asymptotic behavior
     \begin{equation}
\label{PI2asym}
        U(X,T)=\mp (6|X|)^{1/3}\mp \frac{1}{3}6^{2/3}T|X|^{-1/3}
            +\bigO(|X|^{-1}),
            \qquad\mbox{as $X\to\pm\infty$,}
\end{equation}
for each fixed $T\in\mathbb{R}$. The uniqueness of a smooth
solution to (\ref{PI20}) for all $X,T\in\mathbb{R}$ satisfying
(\ref{PI2asym}) follows from earlier results \cite{Menikoff,
Moore}, while the existence has been proven in \cite{CV1}.

\medskip

The aim of this paper is to prove rigorously that the expansion
(\ref{univer0}) holds indeed in the particular case of the small dispersion limit of the KdV equation near
the point of gradient catastrophe for the Hopf equation (\ref{hopf}).
More precisely the solution
$u(x,t,\epsilon)$ of the KdV equation in the neighborhood of
$(x_c,t_c,u_c)$ has an asymptotic expansion as follows,
\begin{equation}
\label{univer}
u(x,t,\e)\simeq u_c +\left(\dfrac{2\epsilon^2}{k^2}\right)^{1/7}
U \left(
\dfrac{x- x_c-6u_c (t-t_c)}{(8k\epsilon^6)^{\frac{1}{7}}},
\dfrac{6(t-t_c)}{(4k^3\epsilon^4)^{\frac{1}{7}}}\right) +O\left( \epsilon^{4/7}\right).
\end{equation}
This expansion holds in the double scaling limit where we let $\epsilon\to 0$ and at the same time $x\to x_c$ and $t\to t_c$ in such a way that
\[\lim\dfrac{x- x_c-6u_c (t-t_c)}{(8k\e^6)^{1/7}}=X, \qquad \lim\dfrac{6(t-t_c)}{(4k^3\epsilon^4)^{1/7}}= T,\]
with $X,T\in\mathbb R$.
The constant $k$ is given by
\[
k=-f_-'''(u_c),
\]
where $f_-$ is the inverse function of the decreasing part of the initial data
$u_0(x)$, which is assumed to be real analytic and with a single negative bump.

\medskip

The universality conjecture of Dubrovin should be seen in comparison
to the known universality results in random matrix theory. For large
unitary random matrix ensembles, local eigenvalue statistics turn
out to be, to some extent, independent of the choice of the ensemble
and independent of the reference point chosen \cite{Deift, DKMVZ2,
DKMVZ1}. Critical break-up times occur when the eigenvalues move
from a one-cut regime to a multi-cut regime. These transitions can
take place in the presence of singular points, of which three
different types are distinguished \cite{DKMVZ2}. Singular interior
points show remarkable similarities with the leading edge of the
oscillatory region for the KdV equation, while singular exterior
points should be compared to the trailing edge of the oscillatory
region. Our focus is on the point of gradient catastrophe, i.e.\ the
break-up point where the oscillations start to set in. This
situation is comparable to a singular edge point in unitary random
matrix ensembles. It was conjectured by Bowick and Br\'ezin, and by
Br\'ezin, Marinari, and Parisi \cite{BB, BMP} that local eigenvalue
statistics in this regime should be given in terms of the Painlev\'e
I hierarchy. In \cite{CV2}, it was proven that indeed double scaling
limits of the local eigenvalue correlation kernel are given in terms
of the Lax pair for the $P_I^2$ equation (\ref{PI20}). In addition,
an expansion similar to (\ref{univer0}) was obtained for the
recurrence coefficients of orthogonal polynomials related to the
relevant random matrix ensembles.

In the setting of PDEs a universality result similar to the KdV case
 has been conjectured  for the semiclassical limit of the focusing nonlinear Schr\"odinger equation \cite{DGK}. In this case the
limiting equations are elliptic and the focusing nonlinear Schr\"odinger equation in the semiclassical limit is considered as a Hamiltonian perturbation of the elliptic system.

\subsection{Statement of result}

Our goal is to find asymptotics as $\e\to 0$ for the solution
$u(x,t,\e)$ of the KdV equation
\[
 u_t+6uu_x+\epsilon^{2}u_{xxx}=0,\quad u(x,0,\e)=u_0(x),
\]
 when $x, t$ are close to the point and time of
gradient catastrophe $x_c,t_c$ for the Hopf equation
\[u_t+6uu_x=0.\]
By the method of characteristics, the solution of the Hopf equation takes the form
\[
u(x,t)=u_0(\xi),\quad x=6tu_0(\xi)+\xi,
\]
so that
\[
u_x(x,t)=\dfrac{u_0'(\xi)}{1+6tu_0'(\xi)},\quad x=6tu_0(\xi)+\xi,
\]
from which one observes that
a gradient catastrophe is reached for
\[
t_c=\dfrac{1}{\max_{\xi\in\mathbb{R}}[-6u'_0(\xi)]}.
\]
We impose the following conditions on the initial data $u_0$.
\begin{Assumptions}\label{assumptions}\
\begin{itemize}
\item[(a)] $u_0(x)$ decays at $\pm\infty$ such that
\begin{equation}
\int_{-\infty}^{+\infty}u_0(x)(1+x^2)dx<\infty,
\end{equation}
\item[(b)] $u_0(x)$ is real analytic and has an analytic continuation to the complex plane in the domain
\[
\mathcal{S}=\{z\in\mathbb{C}: |\Im z|<\tan\theta_0|\Re z|\}\cup\{z\in\mathbb C: |\Im z|<\sigma\}
\]
where $0<\theta_0<\pi/2$ and $\sigma>0$,
\item[(c)] $u_0(x)<0$ and has a single local minimum at a certain point
$x_M$, with
\[u_0'(x_M)=0, \qquad u_0''(x_M)>0,\]$u_0$ is normalized such
that $u_0(x_M)=-1$.
\end{itemize}
\end{Assumptions}
Condition (a) is necessary  to apply the inverse scattering transform \cite{BDT}, while condition (b) is requested in order to get some analyticity properties of the reflection and transmission coefficients for the scattering problem (see below in Section \ref{section: inverse scattering}). Condition (c) is  imposed in order to have the simplest situation in the study of the semiclassical limit of the reflection coefficient of the associated Schr\"odinger equation.

 Let $x_M(t)$ be the $x$-coordinate where $u(x,t)$
reaches the minimum value $-1$. Then for $t<t_c$ we have that
\begin{align}\label{uxx1}
&u_x(x,t)=\dfrac{1}{6t+f'_-(u(x,t))}<0,&\mbox{ for $x<x_M(t)$,}\\
&u_x(x,t)=\dfrac{1}{6t+f'_+(u(x,t))}>0,& \mbox{ for $x>x_M(t)$,}
\end{align}
where $f_{\pm}$ are the inverses of the increasing and decreasing
part of the initial data $u_0(x)$,
\[f_\pm(u_0(x))=x, \qquad f_-(-1,0)=(-\infty, x_M), \qquad f_+(-1,0)=(x_M,+\infty).\]
Since $f_-'$ is negative, it follows that there exists a time $t=t_c$ for which
(\ref{uxx1}) goes to infinity. This happens at the time
\begin{equation}
\label{tc1}
t_c=-\dfrac{1}{6}\max_{\xi\in(-1,0)}f'_-(\xi).
\end{equation}
The above relation shows that the point of gradient catastrophe is characterized also  by
\[
f''_-(u_c)=0.
\]
Since we have assumed that $u_0''(x_M)\neq 0$, it follows that
$x_c<x_M(t_c)$. Summarizing,  the point of gradient catastrophe is
characterized by (\ref{tc1}) and the three equations
\begin{equation}
\label{xc1}
x_c=6t_cu_c+f_-(u_c),\qquad 6t_c+f'_-(u_c)=0,\qquad f_-''(u_c)=0.
\end{equation}
The point of gradient catastrophe is generic if
\begin{equation}
\label{genericity}
f_-'''(u_c)\neq 0.
\end{equation}

\medskip

Our main result is the following.

\begin{theorem}\label{theorem: main}
Let $u_0(x)$ be initial data for the Cauchy problem of the KdV equation
satisfying the conditions described in Assumptions \ref{assumptions}, and satisfying the genericity assumption (\ref{genericity}).
Write $u_c=u(x_c,t_c,0)$, with $x_c$ and $t_c$ the point and
time of gradient catastrophe given by (\ref{xc1}) and (\ref{tc1}).
Now we take a double scaling limit where we let $\epsilon\to
0$ and at the same time we let $x\to x_c$ and $t\to t_c$ in such a way that, for some $X,T\in\mathbb R$,
\begin{equation}
\lim\dfrac{x- x_c-6u_c (t-t_c)}{(8k\e^6)^{1/7}}=X, \qquad \lim\dfrac{6(t-t_c)}{(4k^3\epsilon^4)^{1/7}}= T,
\end{equation}
where
\[
k=-f_-'''(u_c),
\]
and $f_-$ is the inverse function of the decreasing part of the
initial data $u_0(x)$. In this double scaling limit the solution
$u(x,t,\epsilon)$ of the KdV equation (\ref{KdV}) has the following
expansion,
\begin{equation}
\label{expansionu}
u(x,t,\e)=u_c +\left(\dfrac{2\epsilon^2}{k^2}\right)^{1/7}
U \left(
\dfrac{x- x_c-6u_c (t-t_c)}{(8k\epsilon^6)^{\frac{1}{7}}},
\dfrac{6(t-t_c)}{(4k^3\epsilon^4)^{\frac{1}{7}}}\right) +O\left( \epsilon^{4/7}\right).
\end{equation}
Here $U(X,T)$ is the unique real pole-free solution to the $P_I^2$
equation (\ref{PI20}) satisfying the asymptotic condition
(\ref{PI2asym}).
\end{theorem}

\begin{remark}
The expansion (\ref{expansionu}) is very similar to the one
obtained in \cite{CV2} for the recurrence coefficients of
orthogonal polynomials with respect to a weight $e^{-nV(x)}$ on
the real line, in a double scaling limit where the potential $V$
tends, in a double scaling limit, to a critical potential with a
singular edge point.

It should also be noted that the error term in (\ref{expansionu}) is of order $\e^{4/7}$. It
will follow in a nontrivial way from our analysis that the term of order $\e^{3/7}$, which a priori seems to be present when proving Theorem \ref{theorem: main}, vanishes.
\end{remark}

\begin{remark}
If we would consider initial data that do not satisfy the generic condition (\ref{genericity}), our result is not valid any longer. It is likely that the role of the $P_I^2$ solution $U$ would then be taken over by a smooth solution to a higher member of the Painlev\'e I hierarchy.
\end{remark}

The proof of our result goes via the Riemann-Hilbert (RH) approach.
The starting point of our analysis will be the RH problem for the
KdV equation, developed in \cite{BDT, Shabat} using inverse
scattering. The Deift/Zhou steepest descent method \cite{DZ1} has
shown to be a powerful tool in order to obtain asymptotics for
solutions of RH problems. The strategy of this method is to simplify
the RH problem in several steps by applying some invertible
transformations to it. At the end this leads to a RH problem for
which asymptotics can be easily found. The Deift/Zhou steepest
descent method has not only been fruitful in the field of integrable
systems, it lead also to universality results in random matrix
theory \cite{Deift, DKMVZ2, DKMVZ1} and to remarkable combinatorial
results such as in \cite{BDJ}. For the KdV equation, a steepest
descent analysis was carried out by Deift, Venakides, and Zhou
\cite{DVZ, DVZ2}. We will follow the main lines of their approach,
but however with some important modifications which are necessary to
perform the analysis near the point of gradient catastrophe. The
first modification is to perform the steepest descent analysis of
the RH problem associated to the initial data $u_0(x)$ and to the
corresponding reflection coefficient itself, while in \cite{DVZ,
DVZ2} the reflection coefficient was identified with its WKB
approximation before carrying out the RH analysis. Also in the works
by Lax and Levermore \cite{LL} the analysis was performed on some
approximate initial data $\tilde{u}_0(x,\epsilon)$ for which
$\tilde{u}_0(x,\epsilon)\rightarrow u_0(x)$ as $\epsilon\rightarrow
0$, such that the reflection coefficient identifies with its WKB
approximation.
 The second modification concerns the so-called $\mathcal G$-function, which we need to modify in order to have a RH problem that behaves smoothly in the double scaling limit. The third and probably most essential new feature is the construction
of a local parametrix built out of the $\Psi$-functions for the $P_I^2$ equation.

\subsection{Outline for the rest of the paper}
In Section \ref{section: inverse scattering}, we give a short
overview of the inverse scattering approach in order to arrive at
the RH problem for the KdV equation. We also recall some previously
known asymptotic results on the reflection and transmission
coefficients, which will be necessary in Section \ref{section: RH}
when performing the Deift/Zhou steepest descent analysis of the RH
problem. In the asymptotic analysis of the RH problem, we will
construct a $\mathcal G$-function which is slightly modified
compared to the one used in \cite{DVZ, DVZ2}. After the opening of
the lens, the crucial part of Section \ref{section: RH} consists of
the construction of a local parametrix near the critical point
$u_c$. Here we will use a model RH problem associated to the $P_I^2$
equation. Accurate matching of the local parametrix with the
parametrix in the outside region will provide asymptotics for the
solution of the RH problem. In Section \ref{section: proof} finally,
we collect the asymptotic results obtained in the previous section
to prove Theorem \ref{theorem: main}.

\section{Inverse scattering transform}\label{section: inverse scattering}

\subsection{Construction of the RH problem}

In the first part of this section, we recall briefly the construction of the RH problem associated
to the KdV equation, as it was done in e.g.\ \cite{BDT, Shabat}. The RH
problem is constructed by inverse scattering and will be the
starting point of our asymptotic analysis in the next section.
Recall that we consider initial data satisfying Assumptions \ref{assumptions} (a)-(c).

\medskip

The initial value problem for the KdV
equation
\begin{equation}
\label{KdV1} u_t+6uu_x+\epsilon^2 u_{xxx}=0,\quad u(x,0,\e)=u_0(x),
\end{equation}
can be solved by inverse scattering transform \cite{GGKM}.
 Introducing the operators $L$ and $A$, depending on $x$ and also on $t$
through $u=u(x,t)$,
\[
L=\epsilon^2\dfrac{d^2}{dx^2}+u,\quad
A=4\epsilon^2\dfrac{d^3}{dx^3}+3\left(u\dfrac{d}{dx}+\dfrac{d}{dx}u\right),
\]
the KdV equation can be written in the Lax form \cite{Lax}
\begin{equation}
\label{Lax} \dot{L}=[L,A],
\end{equation}
where $\dot{L}=\dfrac{dL}{dt}$ is the operator of multiplication
by $u_t$. The Schr\"odinger equation with potential $u$
\[(\epsilon^2\dfrac{d^2}{dx^2}+u)f=\lambda f,\] can be re-written
 as the first order differential equation
\begin{equation}
\begin{split}
\label{eq1}
&DF=JF+QF,\qquad D=\epsilon\dfrac{d}{dx},\\
&J=\begin{pmatrix}0&1\\\lambda&0\end{pmatrix},\qquad Q\begin{pmatrix}0&0\\-u&0\end{pmatrix},\qquad
F=\begin{pmatrix}f\\g\end{pmatrix}.
\end{split}
\end{equation}
There are two fundamental $2\times 2$ matrix solutions
$\Psi=\Psi(\lambda;x,t,\epsilon)$ and
$\Phi=\Phi(\lambda;x,t,\epsilon)$ to the Schr\"odinger equation
(\ref{eq1}), with
the following normalizations as $x\to \pm \infty$,
\begin{equation}
\label{normalization}
\begin{split}
&\Psi(\lambda;x,t,\epsilon) e^{\frac{i}{\e}(-\lambda)^{1/2}x\sigma_3}\sim\Lambda(\lambda), \qquad \mbox{ as $x\to +\infty$,}\\
&\Phi(\lambda;x,t,\epsilon) e^{\frac{i}{\e}(-\lambda)^{1/2}x\sigma_3}\sim\Lambda(\lambda), \qquad \mbox{ as $x\to -\infty$,}\\
\end{split}
\end{equation}
where
\begin{equation}
\label{Lambda}
 \Lambda(\lambda)=\begin{pmatrix}1&1\\-i(-\lambda)^{1/2}&i(-\lambda)^{1/2}\end{pmatrix},
 \qquad \sigma_3=\begin{pmatrix}1&0\\0&-1\end{pmatrix}.
\end{equation}
Here we take the principal branch of the square root, such that
$(-\lb)^{1/2}$ is analytic in $\mathbb C\setminus [0, +\infty)$
and positive for $\lb<0$.
We denote by  $\Psi_1$, $\Psi_2$, $\Phi_1$, and $\Phi_2$  the columns of the matrices
$\Psi$ and $\Phi$ respectively.
It is known that the vectors $\Psi_1$ and $\Phi_2$ are analytic for $\Im\lambda>0$ and continuous for
$\Im\lambda\geq 0$, and that
$\Psi_2$ and $\Phi_1$ are analytic for $\Im\lambda<0$  and continuous for $\Im\lambda\leq 0$ \cite{Marchenko}.

For analytic initial data satisfying Assumption~\ref{assumptions}(b),  $\Psi(\lambda;x,t,\epsilon)$ and
$\Phi(\lambda;x,t,\epsilon)$ are also analytic for
$(x,\lambda)\in\mathcal{S}\times \Pi_{\theta_0}$ \cite{Fujiie}, where
$\mathcal{S}$ is defined in  Assumption~\ref{assumptions}(b), and
\[
\Pi_{\theta_0}=\{\lambda\in\mathbb{C}\backslash \{0\},\;-2\theta_0+\pi<\arg\lambda<2\theta_0+\pi\}.
\]
Denoting $\Psi(x,\lambda)^*=\overline{\Psi(\bar{x},\bar{\lambda})}$, it follows from the reality of $u_0(x)$ that
\[
\Psi_2=\Psi_1^*,\quad \Phi_2=\Phi_1^*.
\]
Since $\Psi$ and $\Phi$ form two fundamental bases of solutions to the Schr\"odinger equation, they are related to each other by a constant matrix (independent of $x$)
for  $\lambda\in\Pi_{\theta_0}$,
\begin{equation}
\label{transition}
\Psi(\lambda;x,t,\e)=\Phi(\lambda;x,t,\e)\begin{pmatrix}a(\lambda;t,\e)&\bar{b}(\bar{\lambda};t,\e)\\
b(\lb;t,\e)&\bar a(\bar{\lb};t,\e)
\end{pmatrix}.
\end{equation}
The components $a(\lambda;t,\e)$ and $b(\lambda;t,\e)$  (and also $a^*$, $b^*$)
are holomorphic with respect to $\lambda\in\Pi_{\theta_0}$.
Since $\det\Psi\equiv \det\Phi$, we have that
\begin{equation}
\label{normalizationab}
aa^*-bb^*=1,
\end{equation}
and, for real  $\lambda<0$
\[
 |a|^2-|b|^2=1,
\]
which shows that $a\neq 0$ for real $\lambda<0$.
Therefore we can divide by $a$ obtaining
\begin{equation}
\label{transition1} \dfrac{\Psi_{11}}{a}\sim \left\{
\begin{split}
&e^{-\frac{i}{\e}(-\lambda)^{1/2}x}+\dfrac{b}{a}e^{\frac{i}{\e}(-\lambda)^{1/2}x},
& \mbox{ as }x\ra -\infty,\\
&\dfrac{1}{a}e^{-\frac{i}{\e}(-\lambda)^{1/2}x},& \mbox{ as }x\ra +\infty.
\end{split}\right.
\end{equation}
The quantities
\[
r(\lb;t,\e):=\dfrac{b(\lb;t,\e)}{a(\lb;t,\e)},\qquad
t(\lb;t,\e):=\dfrac{1}{a(\lb;t,\e)},
\]
are called reflection and transmission coefficients (from the
left) for the potential $u(x,t,\e)$.  For an arbitrary dependence of
$u(x,t,\e)$ on $t$, it is not possible in general to find the
time-dependence of $a$ and $b$. If however $u(x,t,\e)$ evolves
according to the KdV equation, the Gardner-Greene-Kruskal-Miura
equations \cite{GGKM}
\begin{equation}
\dfrac{da}{dt}=0,\quad
\dfrac{db}{dt}=\frac{8i}{\e}(-\lambda)^{3/2}b,
\end{equation}
hold so that the reflection coefficient evolves according to
\[
r(\lb;t,\e)=r(\lb;0,\e)e^{\frac{8i}{\e}(-\lambda)^{3/2}t}.
\]
In what follows we write $r(\lambda;\e)=r(\lb;0,\e)$ for the
reflection coefficient at time $t=0$.
For $\lambda\rightarrow-\infty$  we have  \cite{Marchenko}
\[
\lim_{\lambda\rightarrow-\infty}r(\lambda;\epsilon)=0.
\]
\begin{remark}
\label{zeros}
For analytic initial data
satisfying Assumptions~\ref{assumptions}, the reflection coefficient
$r(\lb;\e)$ and the transmission coefficient $t(\lambda;\e)$ are meromorphic functions  in the sector $\Pi_{\theta_0}$. The possible poles occur at the zeros of $a(\lambda;\e)$. For initial data satisfying Assumption~\ref{assumptions},
the zeros of $a(\lambda;\e)$ lie in the sector \cite{FujiieRamond}
\[
\pi<\arg\lambda<\pi +2\theta_0.
\]

\end{remark}

\medskip

The time evolution of the scattering data now leads us to the
inverse scattering problem to recover the KdV solution $u(x,t,\e)$
from the reflection coefficient at time $t$.
This brings us to the RH boundary value problem for the initial
value problem of the KdV equation. The $2\times 2$ matrix-valued
function $M$ defined by
\begin{equation}
M(\lb;x,t,\e)=\begin{cases}\begin{array}{ll}\left(\Phi_2(\lb;x,t,\e)e^{-i(-\lb)^{1/2}x}\
\
\dfrac{\Psi_1(\lb;x,t,\e)}{a(\lb;t,\e)}e^{i(-\lb)^{1/2}x}\right),&
\mbox{ as $\lb\in\mathbb C^+$},\\[3ex]
\left(\dfrac{\Psi_2(\lb;x,t,\e)}{a(\lb;t,\e)}e^{-i(-\lb)^{1/2}x}\
\ \Phi_1(\lb;x,t,\e)e^{i(-\lb)^{1/2}x}\right),&\mbox{ as
$\lb\in\mathbb C^-$},
\end{array}
\end{cases}
\end{equation}
where $\Psi_j$ and $\Phi_j$ denote the $j$-th column of $\Psi$ and
$\Phi$, satisfies the following RH conditions \cite{BDT, Shabat}.
\subsubsection*{RH problem for $M$}
\begin{itemize}
\item[(a)] $M(\lb;x,t,\e)$ is analytic for
$\lb\in\mathbb{C}\backslash \mathbb{R}$, \item[(b)] $M$ has continuous boundary values $M_+(\lb)$ and $M_-(\lb)$ when approaching $\lambda\in\mathbb R\setminus\{0\}$ from above and below, and
\begin{align*}&M_+(\lb)=M_-(\lb){\small \begin{pmatrix}1&r(\lb;\e)
e^{2i\alpha(\lambda;x,t)/\e}\\
-\bar r(\lb;\e)e^{-2i\alpha(\lb;x,t)/e}&1-|r(\lb;\e)|^2
\end{pmatrix}},&\mbox{ for $\lb<0$,}\\
&M_+(\lb)=M_-(\lb)\sigma_1,\quad
\sigma_1=\begin{pmatrix}0&1\\1&0\end{pmatrix},&\mbox{ for $\lb>0$},
\end{align*}
\item[(c)] $M(\lb;x,t,\e)\sim
\begin{pmatrix}1&1\\&\\i\sqrt{-\lb}&-i\sqrt{-\lb}\end{pmatrix},\qquad$
for $\lb\rightarrow \infty$.
\end{itemize}
Here $\alpha$ is given by
\[\alpha(\lambda;x,t)=4t(-\lambda)^{3/2}+x(-\lambda)^{1/2}.\]
The solution of the KdV equation can be recovered from the RH
problem by the following formula, see e.g.\ \cite{BDT},
\begin{equation}\label{uM}
u(x,t,\e)=-2i\e\partial_x M_{11}^1(x,t,\e),
\end{equation}
where
$M_{11}(\lb;x,t,\e)=1+\dfrac{M_{11}^1(x,t,\e)}{\sqrt{-\lb}}+\bigO(1/\lb)$
as $\lb\to\infty$.

\subsection{Semiclassical limit of the reflection and transmission coefficients}

In order to study the limit as $\epsilon\rightarrow 0$ of the Riemann-Hilbert problem, it is necessary to obtain semiclassical asymptotics for the reflection coefficient
$r(\lambda;\epsilon)$.

\medskip

Under the validity of Assumptions~\ref{assumptions}, the WKB approximation of
the reflection coefficient as $\e\to 0$ is the following, see \cite{Ramond, Fujiie}.

\begin{itemize}
\item[(i)] For any positive constant $\delta$, we have the following WKB asymptotics
 for $-1+\delta\leq\lambda<0$,
\begin{equation}
\label{WKB1}
\begin{split}
r(\lambda;\e)&= ie^{-\frac{2i}{\e}\rho(\lambda)}(1+\e h_1(\lb;\e)),\\
t(\lambda;\epsilon)&=e^{-\tau(\lambda)/\epsilon}e^{i
(\rho(\lambda)-\widetilde\rho(\lambda))/\epsilon} (1+\e
h_2(\lb;\e)).
\end{split}
\end{equation}
Here $\rho(\lambda)$ and $\widetilde\rho(\lambda)$ are defined by
\begin{equation}\label{def rho}
\rho(\lb)=x_-(\lambda)\sqrt{-\lambda}+\int_{-\infty}^{x_-(\lambda)}[\sqrt{u_0(x)-\lambda}-\sqrt{-\lambda}]dx=
\dfrac{1}{2}\int_{\lambda}^0\dfrac{f_-(u)du}{\sqrt{u-\lambda}},
\end{equation}
\begin{equation}\label{def rho+}
\widetilde\rho(\lb)=x_+(\lambda)\sqrt{-\lambda}-\int^{\infty}_{x_+(\lambda)}[\sqrt{u_0(x)-\lambda}-\sqrt{-\lambda}]dx=\dfrac{1}{2}
\int_{\lambda}^0\dfrac{f_+(u)du}{\sqrt{u-\lambda}},
\end{equation}
where $x_{-}(\lambda)<x_{+}(\lambda)$ are the solutions of the equation
$u_0(x_{\pm}(\lambda))=\lambda$ and $f_{\mp}$ are the inverse function of the decreasing and increasing part of the initial data respectively. The function $\tau$ is defined as
\begin{equation}
\label{tau}
\tau(\lambda)=\int\limits_{x_-(\lambda)}^{x_+(\lambda)}\sqrt{\lambda - u_0(x)}dx=\int_{-1}^{\lambda}\sqrt{\lambda-u}f_+'(u)du+\int_{\lambda}^{-1}\sqrt{\lambda-u}f_-'(u)du.
\end{equation}
The functions $h_1$ and $h_2$ are classical analytic symbols of nonnegative order.\footnote{A function $h(z;\e)$ defined in $U\times (0,\e_0)$ where $U$ is an open set in $\mathbb{C}$, and $\e_0>0$,
 is called a classical analytic symbol  of order $m$ if $h$ is analytic function of $z$ in $U$ and if there is a sequence of analytic functions $a_j(z)$  such that
$h(z;\e)$ admits the series $\sum_{j\geq 0}a_j(z)\e^{j+m}$ as asymptotic expansion as $\e\rightarrow 0$ $\forall z\in U$ and for all compact $K\subset U$ there exists a constant $C>0$ such that
$|a_j(z)|\leq C^{j+1}j^j$  $\forall z\in K$. }

The expansion (\ref{WKB1}) can be extended to  a strip of the complex plane of the form
\begin{equation}\label{definition Pi}
\Pi^+_{\theta_0}=\Pi_{\theta_0}\cap\{-1+\frac{\delta}{2}<\Re(z)<0\}\cap\{0\leq \Im(z)<\tilde{\sigma}\},
\end{equation}
where $0<\tilde{\sigma}<\sigma$, where $\sigma$ has been defined in Assumptions~\ref{assumptions}.
 The constant $\tilde{\sigma}$ has to be chosen sufficiently small
so that  for $\lambda\in \Pi^+_{\theta_0}$  there exist only two simple turning points
$x_{\pm}(\lambda)\in\mathcal{S}$   solving the equation
 $u_0(x_{\pm}(\lambda))= \lambda$.
\item[(ii)] For any fixed  $\delta>0$, there is a constant $c_2>0$ such that for $\lambda\leq -1-\delta$, we have that
\begin{equation}
\label{WKB3}
r(\lambda;\epsilon)=\bigO(e^{-\frac{c_2}{\epsilon}}),\qquad\mbox{ as $\e\to 0$.}
\end{equation}
\item[(iii)] Near $-1$, for sufficiently small $\delta>0$, we have the following asymptotics for $\lambda\in D_+(-1,\delta)=\{\lambda\in\mathbb C:|\lambda+1|<\delta, \;\Im \lambda> 0\}$,
\begin{equation}
\label{WKB2}
r(\lambda;\e)= i \dfrac{ \exp[-2i\rho(\lambda)/\epsilon]}
{N\left(-\dfrac{i\tau(\lambda)}{\pi\epsilon}\right)}(1+O(\e)),\qquad\mbox{ as $\e\to 0$,}
\end{equation}
where
\begin{equation}
\label{N}
N(z)=\dfrac{\sqrt{2\pi}}{\Gamma(1/2+z)}e^{z\log(z/e)},
\end{equation}
and $\Gamma(z)$ is the standard $\Gamma$-function. The functions $\rho(\lambda)$ and $\tau(\lambda)$ are
 the analytic continuations to $D_+(-1,\delta)$
of the corresponding functions defined on the interval $-1\leq \lambda< 0$.
\end{itemize}

\begin{remark}
For complex values of $\lambda$ close to $-1$, the  two simple turning points
$x_{\pm}(\lambda)$ near $x=x_M$ are  $x_{\pm}(\lambda)\approx x_M\pm s\sqrt{\lambda +1}$, with $s^2=2/u_0''(x_M)$. In particular for $\lambda>-1$, these turning points are real, and for $\lambda<-1$,
they are purely imaginary.
Choosing the branch cut of $\rho$ for $\lambda <-1$, we can, for $\delta$ sufficiently small, extend $\rho$ to an analytic function in $D_+(-1,\delta)$.
\end{remark}

\begin{remark}We note that the function
$N(z)$ defined in (\ref{N})
is analytic in $\mathbb C\setminus (-\infty,0]$ and it has poles on the negative axis.
Using properties of the $\Gamma$-function,
one can recover the limits
\begin{align}
&\label{limN1}
\lim_{|z|\rightarrow \infty}N(z)=1,\quad |\arg z|<\pi,\\
&\label{limN2}
\lim_{z\rightarrow 0}N(z)=\sqrt{2}.
\end{align}
The poles of the reflection coefficient $r(\lambda;\e)$ in a neighborhood of $-1$ correspond to
the poles of the function $\Gamma\left(\dfrac{1}{2}-\dfrac{i\tau(\lambda)}{\pi\epsilon}\right)$.
Since $\tau(\lambda)\approx \dfrac{\pi s}{2}(\lambda+1)$ as $\lambda\to -1$, it is clear that the poles of $r(\lambda;\e)$ occur only for $\Im\lambda<0$ (cf.\ Remark \ref{zeros}), although as $\e\to 0$ they are getting closer to $-1$.
\end{remark}

\begin{remark}
We observe that applying the condition (\ref{normalizationab}) we obtain for real $\lambda<0$
\[
|r(\lambda;\e)|^2+|t(\lambda;\e)|^2=1
\]
which shows that
\begin{equation}\label{WKB r}
1-|r(\lambda;\e)|^2=e^{-2\tau(\lambda)}(1+\bigO(\e)),\qquad\mbox{ as
$\e\to 0$,}
\end{equation}
where $\tau$ has been defined in (\ref{tau}). Therefore for
$-1<\lambda<0$, $|r(\lambda;\e)|^2$ is equal to one modulo
exponentially small terms.
\end{remark}
Let us now define
\begin{equation}\label{def kappa1}
\kappa(\lambda;\e)=-r(\lambda;\e)ie^{\frac{2i}{\e}\rho(\lambda)},\qquad\mbox{
as $\lambda\in\mathcal V_+:=\Pi_{\theta_0}^+\cup D_+(-1,\delta)$}.
\end{equation}
where $\rho$ has been defined in (\ref{def rho}).

It then follows that $\kappa(\lambda)\to 1$ uniformly in
$\Pi_{\theta_0}^+$. In $D_+(-1,\delta)$, one obtains  by
(\ref{WKB2}) that
\[
\kappa(\lambda;\e)=\dfrac{1}{N(-\frac{i\tau(\lambda)}{\pi\epsilon})}(1+\bigO(\e)).
\]

For $\lambda\in D_+(-1,\delta)$,
$-i\tau(\lambda)$ stays away from the negative real line, so that by (\ref{limN1}),
$\kappa(\lambda;\e)$ is bounded as $\e\to 0$.

\medskip

Summarizing we can conclude from the above discussion that
\begin{align}
&\label{kappa1}\kappa(\lambda;\e)=1+\bigO(\e), &\mbox{uniformly for $\lambda\in\Pi_{\theta_0}^+$},\\
&\label{kappa3}|\kappa(\lb;\e)|\leq M, &\mbox{as $\lambda\in\mathcal V_+$.}
\end{align}

\begin{Example}
Let us consider the potential $u_0(x)=-1/\mbox{cosh}^2x$.
In this particular case, the Schr\"odinger equation
 \begin{equation}
 \label{cosh}
 \epsilon^2f_{xx}-\frac{1}{\cosh^2 x}f=\lambda f,
 \end{equation}
has a solution of the following form,
\[
f(x;\lambda)=2^{-k}(1-\xi^2)^{k/2}\mbox{\tiny{2}}F\mbox{\!\tiny{1}}(k-s,k+s+1,k+1;\frac{1}{2}(1-\xi)),
\]
where \[\xi=\tanh x, \qquad k=\frac{i}{\e}\sqrt{-\lambda}, \qquad s=-\frac{1}{2}+\frac{i}{\e}\sqrt{1-\frac{\e^2}{4}},\] and
where $\mbox{\tiny{2}}F\mbox{\!\tiny{1}}(a,b,c;z)$ is a hypergeometric function \cite{AS} solving the equation
$z(1-z)w''+[c-(a+b+1)z]w'-abw=0$.
As $x\rightarrow +\infty$, we observe that $\xi\rightarrow 1$, so that
\[
f(x;\lambda,\epsilon)\sim  e^{-\frac{i}{\epsilon}\sqrt{-\lambda}x}\quad x\rightarrow +\infty.
\]
Transforming the hypergeometric function
\begin{align*}
\mbox{\tiny{2}}F\mbox{\!\tiny{1}}(a,b,c;z)&=\dfrac{\Gamma(c)\Gamma(c-a-b)}{\Gamma(c-a)\Gamma(c-b)}\mbox{\tiny{2}}F\mbox{\!\tiny{1}}(a,b,a+b+1-c;1-z)+\\
&\dfrac{\Gamma(c)\Gamma(a+b-c)}{\Gamma(a)\Gamma(b)}(1-z)^{c-a-b}\mbox{\tiny{2}}F\mbox{\!\tiny{1}}(c-a,c-b,c+1-a-b;1-z),
\end{align*}
one obtains  for $x\rightarrow -\infty$ that
\[
{\small f(x;\lambda,\epsilon)\sim\dfrac{\Gamma(k+1)\Gamma(-k)}{\Gamma(s+1)\Gamma(-s)}  e^{\frac{i}{\epsilon}\sqrt{-\lambda}x}+\dfrac{\Gamma(k+1)\Gamma(k)}{\Gamma(k-s)\Gamma(k+s+1)}  e^{-\frac{i}{\epsilon}\sqrt{-\lambda}x}},\ \ \mbox{ as $x\rightarrow -+\infty$}.
\]
Therefore the reflection coefficient is equal to
\begin{equation}
\label{rexample}
r(\lambda;\epsilon)=\dfrac{\Gamma(\frac{i}{\epsilon}\sqrt{-\lambda}-s)\Gamma(\frac{i}{\epsilon}\sqrt{-\lambda}+s+1)\Gamma(-\frac{i}{\epsilon}\sqrt{-\lambda})}{\Gamma(s+1)\Gamma(-s)\Gamma(\frac{i}{\epsilon}\sqrt{-\lambda})}.
\end{equation}
Using well-known asymptotic properties of the $\Gamma$-function \cite{AS}, it is a straightforward calculation to check that (\ref{WKB1}), (\ref{WKB3}), and (\ref{WKB2}) hold for this example.
\end{Example}

\section{Asymptotic analysis of the RH problem near the gradient catastrophe}\label{section: RH}

We will apply the Deift/Zhou steepest descent method on the RH
problem for the KdV equation. The main lines of this method have
been developed in \cite{DVZ, DVZ2}, but the analysis near the time
and point of gradient catastrophe has not been established yet.
The starting point of our analysis is the RH problem for $M$ given
in Section \ref{section: inverse scattering}.

\subsubsection*{RH problem for $M$}
\begin{itemize}
\item[(a)] $M(\lb;x,t,\e)$ is analytic for
$\lb\in\mathbb{C}\backslash \mathbb{R}$, \item[(b)] $M$ satisfies
the following jump conditions,
\begin{align*}&M_+(\lb)=M_-(\lb){\small \begin{pmatrix}1&r(\lb;\e)e^{2i\alpha(\lb;x,t)/\e}\\
-\bar{r}(\lb;\e)
e^{-2i\alpha(\lambda;x,t)/\e}&1-|r(\lb;\e)|^2
\end{pmatrix}},&\mbox{ for $\lb<0$,}\\
&M_+(\lb)=M_-(\lb)\sigma_1,\quad
\sigma_1=\begin{pmatrix}0&1\\1&0\end{pmatrix},&\mbox{ for $\lb>0$},
\end{align*}
\item[(c)] $M(\lb;x,t,\e)\sim
\begin{pmatrix}1&1\\&\\i\sqrt{-\lb}&-i\sqrt{-\lb}\end{pmatrix},\;\;$\quad
for $\lb\rightarrow \infty$.
\end{itemize}

Our goal is to find asymptotics for $M(\lambda;x,t,\e)$ in the small dispersion limit
where $\e\to 0$, while $x, t$ tend at an appropriate rate to
the point and time of gradient catastrophe $x_c,t_c$ for the Hopf
equation $u_t+6uu_x=0.$

\subsection{Construction of the $\mathcal G$-function}

A first crucial issue in the asymptotic analysis is the construction of
an appropriate $\mathcal G$-function. In \cite{DVZ}, a $\mathcal
G$-function has been constructed when taking the small dispersion
limit $\epsilon\to 0$ for fixed $x$ and $t$. However we will,
later on, take a double scaling limit where $x$ and $t$ tend to the
critical values $x_c$ and $t_c$ simultaneously with $\epsilon\to 0$. For this purpose we need to
modify the $\mathcal G$-function in such a way that we will be
able to construct a local parametrix for $\lambda$ near the critical point $u_c$.

\medskip

Let $\mathcal G$ be defined as follows, with $u_c=u(x_c,t_c)$,
\begin{equation}
 \mathcal
G(\lambda;x,t)=\frac{\sqrt{u_c - \lambda}}{\pi}
\int_{u_c}^0\frac{\rho(\eta)-\alpha(\eta;x,t)}{(\eta
-\lambda)\sqrt{\eta-u_c}}d\eta,\label{def G}
\end{equation}
where $\rho$ and $\alpha$ are defined as before by
\begin{align*}&\rho(\lb)=x_-(\lambda)\sqrt{-\lambda}+\int_{-\infty}^{x_-(\lb)}[\sqrt{u_0(x)-\lambda}-\sqrt{-\lambda}]dx,\\
&\alpha(\lb;x,t)=4t(-\lb)^{3/2}+x(-\lb)^{1/2},\end{align*} and
$\sqrt{u_c-\lambda}$ is analytic off $[u_c,+\infty)$, positive for
$\lambda<u_c$.
 Now $\mathcal G$ is
analytic in $\mathbb C\setminus [u_c,+\infty)$ and $\mathcal
G(\lambda)=\bigO(\lambda^{-1/2})$ as $\lambda\to\infty$. One checks
directly that
\begin{eqnarray}\label{g1}
\mathcal{G}_1(x,t):=\lim_{\lambda\to \infty}(-\lambda)^{1/2}\mathcal
G(\lambda;x,t)=\frac{1}{\pi}
\int_{u_c}^0\dfrac{\rho(\eta)-\alpha(\eta;x,t)}{\sqrt{\eta-u_c}}d\eta.
\end{eqnarray}
Since $\rho(\eta)$ is independent of $x$, we obtain
\begin{eqnarray}
\partial_x \mathcal{G}_1(x,t)&=&\frac{-1}{\pi}\int_{u_c}^0\sqrt{\frac{-\eta}{\eta-u_c}}d\eta\nonumber\\
&=&\frac{u_c}{2}.\label{dg1}
\end{eqnarray}

In the following proposition, we collect two
boundary value relations for the $\mathcal G$-function. Together with the asymptotic behavior of $\mathcal G$, we could have considered those as the defining properties for the $\mathcal G$ as well, instead of defining it explicitly by means of (\ref{def G}) and deriving the properties afterwards.

\begin{proposition}\label{prop g}
$\mathcal G$ is analytic in $\mathbb
C\setminus[u_c,+\infty)$, and we have
\begin{align*}
&\mathcal G_+(\lambda)+\mathcal G_-(\lambda)=0,&\mbox{ for $\lambda\in(0,+\infty)$,}\\
&\mathcal G_+(\lambda)+\mathcal G_-(\lb)-2\rho(\lambda)+2\alpha(\lambda)=0,&\mbox{ for
$\lambda\in(u_c,0)$},\\
&
\mathcal G_+(\lambda)-\mathcal G_-(\lambda)=0,&\mbox{ for $\lambda\in(-\infty,u_c)$}.
\end{align*}
\end{proposition}
\begin{proof}
The first statement is immediate since
$\sqrt{u_c-\lambda}_+=-\sqrt{u_c-\lambda}_-$ and since $\lb>0$
does not belong to the interval of integration. For the second
statement, note that as $\lambda\in (u_c,0)$, $\mathcal
G_++\mathcal G_-$ is equal to a contour integral around $\lb$.
A simple residue argument gives us the required result. The last statement is obvious because $\sqrt{u_c-\lambda}$ is branched for $\lambda>u_c$ and $\lb<u_c$
does not belong to the interval of integration.
\end{proof}

\subsection{First transformation of the RH problem}
In this section we perform a first transformation of the RH
problem for $M$. The goal of this transformation is to simplify
the jump matrices for the RH problem, using the properties of the
$\mathcal G$-function in an appropriate way.

\medskip

We define
\begin{equation}\label{def T}
T(\lambda;x,t,\e)=M(\lambda;x,t,\e)e^{-\frac{i}{\epsilon}\mathcal
G(\lambda;x,t)\sigma_3}.
\end{equation}
From the RH conditions for $M$ and the definition of $T$, we can now derive the RH conditions for $T$.

\subsubsection*{RH problem for $T$}
\begin{itemize}
\item[(a)] $T$ is analytic in $\mathbb C\setminus \mathbb R$,
\item[(b)] $T_+(\lambda)=T_-(\lambda)v_T(\lambda)$, \ \ \ as
$\lambda\in\mathbb R$,
\[v_T(\lambda)=e^{\frac{i}{\epsilon}\mathcal
G_-(\lambda)\sigma_3}v_M(\lambda)
e^{-\frac{i}{\epsilon}\mathcal G_+(\lambda)\sigma_3},\]
\item[(c)]
$T(\lambda)\sim\begin{pmatrix}1&1\\
i\sqrt{-\lambda}&-i\sqrt{-\lambda}\end{pmatrix}$ as
$\lambda\to\infty$.
\end{itemize}

Using Proposition \ref{prop g}, one checks that the jump matrix $v_T$
transforms to
\begin{equation}
v_T(\lambda)=\begin{cases}
\begin{array}{ll}
\sigma_1, &\mbox{ as $\lambda>0$,}\\[0.8ex]
\begin{pmatrix}1&r(\lb)e^{\frac{2i}{\e}(\mathcal
G(\lb)+\alpha(\lb))}\\
-\bar r(\lb)e^{-\frac{2i}{\e}(\mathcal G(\lb)+\alpha(\lb))} &1-|r(\lb)|^2
\end{pmatrix},&\mbox{ as $\lambda\in (-\infty,u_c)$,}\\[3ex]
\begin{pmatrix}e^{-\frac{i}{\e}(\mathcal G_+(\lb)-\mathcal G_-(\lb))}
&r(\lb)e^{\frac{2i}{\e}\rho(\lb)}\\
-\bar{r}(\lb)e^{-\frac{2i}{\e}\rho(\lb)}&(1-|r(\lb)|^2)e^{\frac{i}{\e}(\mathcal
G_+(\lb)-\mathcal G_-(\lb))}
\end{pmatrix},&\mbox{ as $\lambda\in (u_c,0)$,}
\end{array}
\end{cases}
\end{equation}
It is important that we can express the KdV solution $u(x,t,\e)$ in terms of the new RH problem for $T$.
Using (\ref{uM}), (\ref{dg1}), and (\ref{def T}), the solution of
the KdV equation can now be obtained from
\begin{eqnarray}
u(x,t,\e)&=&2\partial_x \mathcal{G}_1(x,t) -2i\e\partial_x T_{11}^1(x,t,\e)\nonumber \\
&=&u_c-2i\e\partial_x T_{11}^1(x,t,\e),\label{uT}
\end{eqnarray}
where $T_{11}^1$ is given by
\[T_{11}(\lb;x,t,\e)=1+
\frac{T_{11}^1(x,t,\e)}{\sqrt{-\lb}}+O(\lb^{-1}),\qquad \mbox{ as $\lb\to\infty$}.\]

\medskip

Let us define an auxiliary function $\phi$ in $\mathcal V=\mathcal V_+\cup\mathcal V_-$,
 where $\mathcal V_+$ is the region defined in (\ref{def kappa1}) and (\ref{definition Pi}), and $\mathcal V_-=\overline{\mathcal V_+}$, by
\begin{align}
&\label{definition phi}\phi(\lb;x,t)=\mathcal G(\lb;x,t)-\rho(\lb)+\alpha(\lb;x,t),&\mbox{ for $\lb\in\mathcal V_+$,}\\
&\label{definition phi2}\phi(\lb;x,t)=\phi^*(\lb;x,t)=\overline{\phi(\bar\lambda;x,t)},&\mbox{ for $\lb\in\mathcal V_-$.}
\end{align}
Note that because of reality, $\phi$ is analytic across $(-1,u_c)$,
but not on $(-1-\delta,-1)$, because this is a part of the branch
cut for $\rho$, as we discussed in Section \ref{section: inverse
scattering}. On $(u_c,0)$, $\phi$ is not analytic because of the
branch cut for $\mathcal G$. We obtain from Proposition \ref{prop g}
that the following equation holds,
\begin{equation}2\phi_{+}(\lb;x,t)=\mathcal G_+(\lb;x,t)-\mathcal
G_-(\lb;x,t),\qquad \text {for $\lb\in (u_c,0)$}.\label{def phi 2}\end{equation}
\begin{lemma}
The function $\phi$ defined in (\ref{definition
phi})-(\ref{definition phi2}) takes the form
\begin{equation}
\label{phi}
\phi(\lambda;x,t)=\sqrt{u_c-\lambda}(x-x_c-6u_c(t-t_c))+\int_{\lambda}^{u_c}(f_-'(\xi)+6t)\sqrt{\xi-\lambda}d\xi.
\end{equation}
\end{lemma}
\begin{proof}
Using the definition (\ref{def G}) of $\mathcal G$ and (\ref{def rho}), it follows that
\[
\mathcal{G}(\lambda)=\frac{\sqrt{u_c - \lambda}}{\pi}
\int_{u_c}^0\frac{ \frac{1}{2}\int_{\eta}^0\frac{f_-(\xi)d\xi}{\sqrt{\xi-\eta}}-4t(-\eta)^{3/2}-x(-\eta)^{1/2}}{(\eta
-\lambda)\sqrt{\eta-u_c}}d\eta.
\]
For $\lambda\notin(u_c,0)$ one uses the residue theorem to integrate
the terms in $x$ and $t$ while for the double integral we exchange
the order of integration obtaining
\begin{equation*}
\begin{split}
\mathcal{G}(\lambda)&=\sqrt{u_c-\lambda}(x-6u_ct)+4t(u_c-\lambda)^{3/2}-\alpha(\lambda;x,t)\\
&+\frac{\sqrt{u_c - \lambda}}{2\pi}
\int_{u_c}^0f_-(\xi)\int_{u_c}^{\xi}\dfrac{
d\eta}{\sqrt{\xi-\eta}(\eta -\lambda)\sqrt{\eta-u_c}}d\xi.
\end{split}
\end{equation*}
Applying the residue theorem to the last integral one obtains
\begin{equation*}
\begin{split}
\mathcal{G}(\lambda)&=\sqrt{u_c-\lambda}(x-6u_ct)+4t(u_c-\lambda)^{3/2}-\alpha(\lambda;x,t)\\
&+\frac{1}{2}\int_{\lambda}^0\dfrac{f_-(\xi)d\xi}{\sqrt{\xi-\lambda}}++\frac{1}{2}\int^{\lambda}_{u_c}
\dfrac{f_-(\xi)d\xi}{\sqrt{\xi-\lambda}}.
\end{split}
\end{equation*}
Integrating by parts the last term of the l.h.s.\ of the above relation
 and using the identity $x_c=6u_ct_c+f_-(u_c)$, one arrives at the statement.
 In the case   $\lambda\in (u_c,0)$ one has to evaluate the principal value of the integral defining $\mathcal{G}$
in the same way as done above.
\end{proof}
The introduction of $\phi$ enables us to rewrite the jump matrix $v_T$ in a more convenient form.

\medskip

Indeed we have that
\begin{equation}\label{vT2}
v_T(\lambda)=\begin{cases}
\begin{array}{ll}
\sigma_1, &\mbox{ as $\lambda>0$,}\\[0.8ex]
{\small \begin{pmatrix}1&r(\lb)e^{\frac{2i}{\e}(\mathcal
G(\lb)+\alpha(\lb))}\\
-\bar r(\lb)e^{-\frac{2i}{\e}(\mathcal G(\lb)+\alpha(\lb))} &1-|r(\lb)|^2
\end{pmatrix}},&\mbox{ as $\lambda<-1-\delta$,}\\[3ex]
\begin{pmatrix}1&i\kappa_+(\lb)e^{\frac{2i}{\e}\phi_+(\lb)}\\
i\kappa_-^*(\lb)e^{-\frac{2i}{\e}\phi_-(\lb)} &1-|r(\lb)|^2
\end{pmatrix},&\mbox{ as $\lambda\in (-1-\delta,u_c)$,}\\[3ex]
\begin{pmatrix}e^{-\frac{2i}{\e}\phi_+(\lb)}
&i\kappa(\lb)\\
i\bar\kappa(\lb)&(1-|r(\lb)|^2)e^{\frac{2i}{\e}\phi_+(\lb)}
\end{pmatrix},&\mbox{ as $\lambda\in (u_c,0)$.}
\end{array}
\end{cases}
\end{equation}
As before we have written $\kappa$ for
\begin{equation}\label{def kappa}
\kappa(\lambda;\e)=-ir(\lb;\e)e^{\frac{2i}{\e}\rho(\lb)},\qquad \mbox{ as $\lambda\in\mathcal V_+$},
\end{equation}
with boundary values on $\mathbb R$ denoted by $\kappa_+$.
We have written $\kappa_-^*(\lb)=\bar\kappa_+(\lambda)$.
In particular near $u_c$, where our main focus is, the introduction of $\phi$
will turn out to be convenient.

\medskip

In the following proposition, we discuss the behavior of $\phi$ in different regions of $\mathcal V$.

\begin{proposition}\label{prop phi}
For sufficiently small $\delta>0$, there exists $\delta_1>0$ and there exists a neighborhood $\mathcal W$ of $(-1-\delta,u_c-\delta)$ such that
for
$|x-x_c|<\delta_1$, $|t-t_c|<\delta_1$, the following holds,
\begin{equation}
\label{phiinequalities}
\begin{array}{ll}
\Im\phi(\lb;x,t)>0,&\mbox{ as $\lb\in\mathcal W\cap\{z\in\mathbb C:\Im z>0\}$,}\\[1.5ex]
\Im\phi(\lb;x,t)<0,&\mbox{ as $\lb\in\mathcal W\cap\{z\in\mathbb C:\Im z<-0\}$,}\\[1.5ex]
\Im\phi_+(\lb;x,t)<0,&\mbox{ as $\lb\in [u_c+\delta ,0]$,}\\[1.5ex]
-\tau(\lambda)+i\phi_+(\lb;x,t)<0,&\mbox{ as $\lb\in [u_c+\delta,0]$.}
\end{array}
\end{equation}
\end{proposition}
\begin{proof}
Taking the derivative of (\ref{phi}) one obtains
\begin{equation}
\label{phiprime}
\phi'(\lambda;x,t)=-\dfrac{1}{2\sqrt{u_c-\lambda}}(x-x_c-6u_c(t-t_c))-\dfrac{1}{2}
\int_\lambda^{u_c}\dfrac{f_-'(\xi)+6t}{\sqrt{\xi-\lambda}}d\xi.
\end{equation}
Since $f_-'(\xi)+6t_c<0$ for $\xi\in (-1,0)$, it follows that
\[
\phi'(\lambda;x_c,t_c)>0,\qquad\mbox{ for $\lambda\in[-1,u_c-\delta]$.}
\]
By continuity there exists $\delta_1>0$ such that for $|x-x_c|<\delta_1,\;\;|t-t_c|<\delta_1$,
\begin{equation}
\phi'(\lambda;x,t)>0,\qquad \mbox{ as }\lambda \in [-1,u_c-\delta].
\end{equation}
From the above inequality, by the
Cauchy-Riemann relations, it follows that
\begin{align}&\Im\phi(\lb;x,t)>0, &\mbox{ for $\lb\in\mathcal W\cap\{\lb\in\mathbb C:\Re \lb>-1, \Im \lb>0\}$},\\
&\Im\phi(\lb;x,t)<0,&\mbox{ for
$\lb\in\mathcal W\cap\{\lb\in\mathbb C:\Re \lb>-1, \Im \lb<0\}$,}
\end{align}
 at least if the neighborhood $\mathcal W$ is chosen sufficiently small.

\medskip

In the case $\lambda<-1$ close to $-1$, we need to be a little bit more careful because of the branch cut for $\phi$.
The function $f_-$, the inverse of the decreasing part of the initial data, now assumes the form
$$f_-(\lambda)=x_M-\dfrac{\sqrt{\lambda +1}}{\sqrt{u_0''(x_M)/2}}(1+\bigO(\lambda+1)),\qquad\mbox{ as $\lb \to -1$}.$$
Since the function $\sqrt{\lambda +1}$ is analytic
in $\mathbb{C}\backslash (-\infty, -1]$ and positive for $\lambda>-1$, it follows
 from (\ref{phi}) that
 \begin{align}&\Im\phi_+(\lambda;x,t)>0,&\mbox{ for $\lambda\in[-1-\delta,-1)$,}\\
&\Im\phi_-(\lambda;x,t)<0,&\mbox{ for $\lambda\in[-1-\delta,-1)$,}
\end{align}
Since $\Re\phi_{\pm}'(\lambda;x,t)>0$ for $\lambda\in[-1-\delta,-1]$, by the Cauchy-Riemann relations it again follows that the first and second inequalities in (\ref{phiinequalities}) hold true.

\medskip

Regarding  the third inequality in  (\ref{phiinequalities})  we have for
$\lb\in
[u_c+\delta,0]$ that
\[
i\phi_+(\lambda;x,t)=\sqrt{\lambda -
u_c}(x-x_c-6u_c(t-t_c))-\int_{u_c}^{\lambda}
(f'_-(\xi)+6t)\sqrt{\lambda-\xi}d\xi.
\]
This means that $i\phi_+(\lambda;x_c,t_c)>0$ for $\lambda\in [u_c+\delta,0]$. When $|x-x_c|<\delta_1,\;\;|t-t_c|<\delta_1$ with $\delta_1>0$ sufficiently small, it follows that
$i\phi_+(\lambda;x,t)>0$ for  $\lambda\in [u_c+\delta,0]$ .

\medskip

For proving the last inequality of Proposition~\ref{prop phi}, with $\tau$ defined as in (\ref{tau}), it is straightforward to verify that for $\lb\in
[u_c+\delta,0]$,
\begin{align*}
-\tau(\lambda)+i\phi_+(\lb;x,t)=&\sqrt{\lambda-u_c}[x-x_c-6u_c(t-t_c)]
-4t(\lambda-u_c)^{\frac{3}{2}}\\
&+\int_{-1}^{u_c}\sqrt{\lambda-\xi}f'_-(\xi)d\xi
-\int_{-1}^{\lambda}\sqrt{\lambda-\xi}f'_+(\xi)d\xi,
\end{align*}
where  $f_+$ is the inverse function of the increasing part of the initial data $u_0(x)$.
Since \[-\tau(\lambda;x_c,t_c)+i\phi_+(\lb;x_c,t_c)<0,\qquad\mbox{ as $\lb\in [u_c+\delta,0]$,}\]
there exists $\delta_1>0$  sufficiently small, such that
$-\tau(\lambda;x,t)+i\phi_+(\lb;x,t)<0$ when  $|x-x_c|<\delta_1$ and
$|t-t_c|<\delta_1$.
\end{proof}

As $\e\to 0$, the above proposition shows us exponential decay of
the diagonal entries of $v_T$ on $[u_c+\delta,0]$, and oscillating
behavior of the off-diagonal entries of $v_T$ on $(-1-\delta,u_c)$.
Performing a next transformation will enable us to deform the
oscillatory entries to exponentially decaying entries as well.

\subsection{Opening of the lens}

We are able to factorize the jump matrix in the following way for $-1-\delta<\lambda<u_c$,
\begin{equation}
v_T(\lb)=\begin{pmatrix}1&0\\
i{\kappa}^*_-(\lb)e^{-\frac{2i}{\e}
\phi_-(\lb)}&1
\end{pmatrix}\begin{pmatrix}1&i\kappa_+(\lb)e^{\frac{2i}{\e}\phi_+(\lb)}\\
0&1
\end{pmatrix}.
\end{equation}
Furthermore the first factor can be extended analytically to the lower half plane and the second factor can be extended to the upper half plane. This observation enables us to move the jump contour, which coincided with the real line so far, into the complex plane. The spirit of the Deift/Zhou steepest descent method is that, deforming the contours, one can deform oscillatory jump matrices (on the real line) to exponentially decaying jump matrices (in the complex plane).

\medskip

We will use the factorization
of the jump matrix to open lenses along the interval
$(-1-\delta,u_c)$ for some sufficiently small but fixed
$\delta>0$. It is not necessary to open lenses elsewhere on the real line, because there the jump matrices will turn out to be exponentially small already without deforming the contour. However it is necessary to open the lens starting from some point slightly to the left of $-1$. Opening the lens exactly at $-1$ would not lead to exponentially small jump matrices near $-1$.
Let us consider a lens-shaped region as shown in
Figure \ref{figure: sigmaS}. We retain the freedom to specify the
precise choice of the lens later on, bur for now we assume that $\Sigma_1$ is, except near $u_c$, contained in $\mathcal W\cap\{z\in\mathbb C:\Im z>0\}$, and that $\Sigma_2$ is, except near $u_c$, contained in $\mathcal W\cap\{z\in\mathbb C:\Im z<0\}$, where $\mathcal W$ is a region for which Proposition \ref{prop phi} is valid.

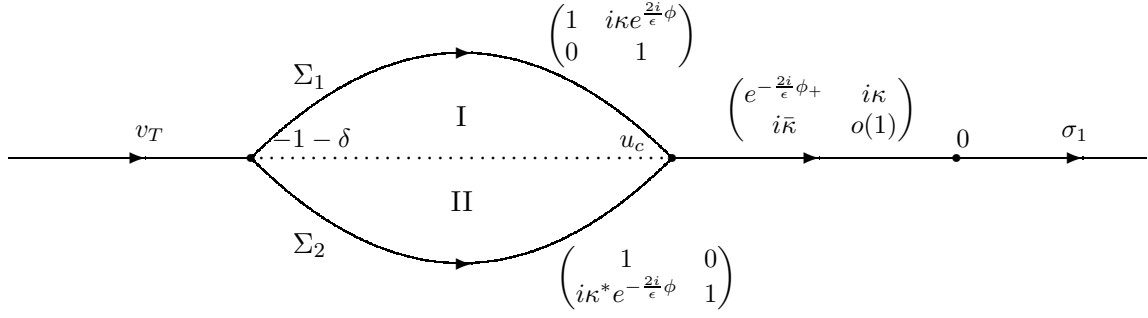
\begin{figure}[t]
\begin{center}
    \setlength{\unitlength}{1.4mm}
    \begin{picture}(137.5,26)(22,11.5)
        \put(112,25){\thicklines\circle*{.8}}
        \put(45,25){\thicklines\circle*{.8}}
        \put(47,26){\small $-1-\delta$}
        \put(112,26){\small $0$}\put(122,27){\small $\sigma_1$}
        \put(85,25){\thicklines\circle*{.8}} \put(80,26){\small $u_c$}
        \put(99,25){\thicklines\vector(1,0){.0001}}
        \put(85,25){\line(1,0){45}}
        \put(124,25){\thicklines\vector(1,0){.0001}}
        \put(22,25){\line(1,0){23}}
        \put(35,25){\thicklines\vector(1,0){.0001}}
\put(74,13){\small
$\begin{pmatrix}1&0\\i\kappa^* e^{-\frac{2i}{\epsilon}\phi}&1\end{pmatrix}$}
        \qbezier(45,25)(65,45)(85,25) \put(66,35){\thicklines\vector(1,0){.0001}}
        \qbezier(45,25)(65,5)(85,25) \put(66,15){\thicklines\vector(1,0){.0001}}
\put(73,36){\small
$\begin{pmatrix}1&i\kappa e^{\frac{2i}{\epsilon}\phi}\\0&1\end{pmatrix}$}
\put(34,27){\small $v_T$} \put(90,29){\small $\begin{pmatrix}
e^{-\frac{2i}{\epsilon}\phi_+}&i\kappa\\i\bar\kappa& o(1)
\end{pmatrix}$}
\put(64.5,28){I} \put(64,20){II} \put(49,32){$\Sigma_1$}
\put(49,16){$\Sigma_2$} \multiput(45,25)(1,0){40}{\circle*{0.1}}
    \end{picture}
    \caption{The jump contour $\Sigma_S$ and the jumps for $S$}
    \label{figure: sigmaS}
\end{center}
\end{figure}

\medskip

Define $S$ as follows,
\begin{equation}
S(\lambda)=\begin{cases}\begin{array}{ll}
T(\lambda)\begin{pmatrix}1&-i\kappa(\lb)e^{\frac{2i}{\e}\phi(\lb)}\\
0&1
\end{pmatrix},&\mbox{ in region I,}\\[3ex]
T(\lambda)\begin{pmatrix}1&0\\
i\kappa^*(\lb)e^{-\frac{2i}{\e}\phi(\lb)}&1
\end{pmatrix},&\mbox{ in region II},\\[3ex]
 T(\lambda), &\mbox{ elsewhere},
\end{array}
\end{cases}
\end{equation}
with $\kappa^*(\lb)=\bar{\kappa}(\bar{\lambda})$.

Now the RH problem for $S$ takes the following form.
\subsubsection*{RH problem for $S$}
\begin{itemize}
\item[(a)] $S$ is analytic in $\mathbb C\setminus \Sigma_S$,
\item[(b)] $S_+(\lambda)=S_-(\lambda)v_S$ for $\lambda\in\Sigma_S$,
with
\begin{equation}\label{vS}
v_S(\lambda)=\begin{cases}
\begin{array}{lr}
\begin{pmatrix}1& i\kappa(\lb)e^{\frac{2i}{\e}\phi(\lb)}\\
0&1
\end{pmatrix},&\mbox{ on $\Sigma_1$},\\[3ex]
 \begin{pmatrix}1&0\\
i\kappa^*(\lb)e^{-\frac{2i}{\e}\phi(\lb)}&1
\end{pmatrix},&\mbox{ on $\Sigma_2$,}\\[3ex]
\begin{pmatrix}e^{-\frac{2i}{\e}\phi_+(\lb)}
&i\kappa(\lb)\\
i\bar\kappa(\lb)&(1-|r(\lb)|^2)e^{\frac{2i}{\e}\phi_+(\lb)}
\end{pmatrix},&\mbox{ as $\lambda\in (u_c,0)$,}\\[3ex]
v_T(\lambda),&\mbox{\hspace{-3cm} as
$\lambda\in(-\infty,-1-\delta)\cup(0,+\infty)$.}
\end{array}
\end{cases}
\end{equation}
\item[(c)] $S(\lambda)\sim \begin{pmatrix}1&1\\
i\sqrt{-\lambda}&-i\sqrt{-\lambda}\end{pmatrix}$ as
$\lambda\to\infty$.
\end{itemize}
Since $S(\lb)=T(\lb)$ for large $\lb$, formula (\ref{uT})
remains valid for $S$, so that we can retrieve the KdV solution by
\begin{equation}
u(x,t,\e)=u_c-2i\e\partial_x S_{11}^1(x,t,\e),\label{uS}
\end{equation}
where
\begin{equation}S_{11}(\lb;x,t,\e)=1+
\frac{S_{11}^1(x,t,\e)}{\sqrt{-\lb}}+O(\lb^{-1}),\qquad \mbox{ as
$\lb\to\infty$}.\end{equation}

In the following proposition, we show that the jump matrix $v_S$
converges to a constant matrix uniformly fast as $\e\to 0$, except
for in a neighborhood of $u_c$. We should note that this is
only true because of the semiclassical limit of the reflection coefficient described in Section \ref{section: inverse scattering}, so indirectly this relies on Assumptions \ref{assumptions}.

\begin{proposition}\label{prop vs}
We can choose a suitable contour $\Sigma_S$ such that for any neighborhood $\mathcal U$ of $u_c$,
\begin{equation}
v_S(\lb)=v^{(\infty)}(\lb)(I+\bigO(\e)), \qquad\mbox{as
$\e\to 0$,}
\end{equation}
uniformly for $\lb\in\Sigma_S\setminus \mathcal U$, with
$v^{(\infty)}$ defined by
\begin{equation}
v^{(\infty)}(\lb)=\begin{cases}\begin{array}{ll}
\sigma_1,&\mbox{for $\lb>0$,}\\
i\sigma_1,&\mbox{for $\lb\in(u_c,0)$,}\\ I,&\mbox{elsewhere.}
\end{array}
\end{cases}
\end{equation}
\end{proposition}
\begin{proof}
\begin{itemize}
\item[(i)] For $\lb>0$, $v_S(\lb)=v^{(\infty)}(\lb)$ so that the
result holds trivially here.
\item[(ii)]
For
$\lb\in(u_c,0)\setminus \mathcal U$, it follows from
(\ref{kappa1}) that that the off-diagonal entries of
$ v_S(\lb)$  tend to $i$ with uniformly small error of order $\e$.
The $11$-entry and the $22$-entry of $ v_S(\lb)$
are exponentially small because of (\ref{tau}) and Proposition \ref{prop phi}.
\item[(iii)] For $\lb\in(\Sigma_1\cup\Sigma_2)\setminus \mathcal U$, under the assumption of a well chosen contour in view of
Proposition \ref{prop phi}, it follows from (\ref{kappa3}) and Proposition \ref{prop phi} that
\[v_S(\lb)=v^{(\infty)}(\lb)(I+\bigO(e^{-\frac{c_3}{\e}})), \qquad\mbox{as
$\e\to 0$,}\]
for some constant $c_3>0$.
\item[(iv)] For $\lb<-1-\delta$, $v_S(\lb)\to I$ because of the
uniform convergence of the reflection coefficient, see
(\ref{WKB3}), together with the reality of $\mathcal G$ and $\alpha$.
\end{itemize}
\end{proof}
\begin{remark}
Note that the error term of the jump matrices is exponentially small everywhere except for on the interval $(u_c,0)$, where the WKB approximation for the reflection coefficient causes the $\bigO(\e)$-error.
\end{remark}

\subsection{Outside parametrix}
Ignoring the exponentially small jumps and a small neighborhood
$\mathcal U$ of $u_c$ where the uniform exponential decay of the jump matrices does not
remain valid, our RH problem reduces to the following RH problem for
$P^{(\infty)}$.

\subsubsection*{RH problem for $P^{(\infty)}$}
\begin{itemize}
\item[(a)] $P^{(\infty)}:\mathbb C\setminus [u_c, +\infty) \to
\mathbb C^{2\times 2}$ is analytic, \item[(b)] $P^{(\infty)}$
satisfies the following jump conditions on $(u_c, +\infty)$,
\begin{align}
&P_+^{(\infty)}=P_-^{(\infty)}\sigma_1, &\mbox{ as $\lambda\in (0, +\infty)$},\\
&P_+^{(\infty)}= iP_-^{(\infty)}\sigma_1, &\mbox{ as $\lambda\in
(u_c,0)$},\label{RHP Pinfty b}
\end{align}
\item[(c)]$P^{(\infty)}$ has the following behavior as $\lambda\to\infty$,
\begin{equation}P^{(\infty)}(\lambda)\sim
\begin{pmatrix}1&1\\ i(-\lambda)^{1/2} & -i(-\lambda)^{1/2}\end{pmatrix}
. \label{RHP Pinfty c}\end{equation}
\end{itemize}

\medskip

We can construct the solution of this RH problem explicitly as
follows,
\begin{equation}
\label{def Pinfty}
P^{(\infty)}(\lambda)=(-\lambda)^{1/4}(u_c-\lambda)^{-\sigma_3
/4}\begin{pmatrix}1&1\\ i& -i
\end{pmatrix},
\end{equation}
from which it follows directly that the asymptotic condition (\ref{RHP
Pinfty c}) can be strengthened to
\begin{equation}\label{RHP Pinfty c2}P^{(\infty)}(\lambda)=\left(I+\frac{u_c}{4\lambda}\sigma_3+\bigO(\lambda^{-2})\right)
\begin{pmatrix}1&1\\ i(-\lambda)^{1/2} & -i(-\lambda)^{1/2}\end{pmatrix},
\quad\mbox{as $\lb\to\infty$.}\end{equation}

In the outside region away from $u_c$,
the leading order asymptotics of $S$ will be determined
by $P^{(\infty)}$. To obtain uniform asymptotics and asymptotics beyond the leading term,
we still need to construct a local parametrix near $u_c$ which matches with the outside parametrix at $\partial\mathcal U$.
This is the goal of the next section. Note also that the uniform decay of the jump matrices remains valid near $-1$ and $0$, so that there is no necessity to construct local parametrices near those points.

\subsection{Local parametrix near $u_c$}

The uniform convergence of the jump matrices as $\epsilon\to 0$
breaks down near $u_c$. Here we need to construct a local parametrix
by mapping a suitable model RH problem onto a neighborhood
$\mathcal U$ of $u_c$. We will use a model RH problem which is associated
with a fourth order analogue of the Painlev\'e I equation, the $P_I^2$ equation (\ref{PI20}). The aim of this section is
to construct a local parametrix $P$ in $\mathcal U$ which has
approximately the same jumps as $S$ has in $\mathcal U$, and which
'matches' with $P^{(\infty)}$ at $\partial \mathcal U$. With
'matching' we mean that $P(\lambda)P^{(\infty)}(\lambda)^{-1}$
tends to the identity matrix in a suitable double scaling limit. The appropriate double scaling limit
will turn out to be the one where we let $\epsilon\to 0$ and in the same time we let $x\to
x_c$ and $t\to t_c$ in such a way that
$t-t_c=\bigO(\epsilon^{4/7})$ and
$x-x_c-6u_c(t-t_c)=\bigO(\epsilon^{6/7})$. More precisely, we want
$P$ to satisfy a RH problem of the following form.
\subsubsection*{RH problem for $P$}
\begin{itemize}
\item[(a)]$P:\overline{\mathcal U}\setminus \Sigma_S\to\mathbb C^{2\times 2}$ is
analytic,
\item[(b)]$P$ satisfies the following jump condition on $\mathcal U \cap
\Sigma_S$,
\begin{equation}\label{RHP
P:b}P_+(\lambda)=P_-(\lambda)v_P(\lambda),
\end{equation}
with $v_P$ given by
\begin{equation}\label{vP}
v_{P}(\lb)=\begin{cases}
\begin{array}{ll}
\begin{pmatrix}1&ie^{\frac{2i}{\e}\phi(\lb;x,t)}\\
0&1
\end{pmatrix},&\mbox{ as $\lb\in\Sigma_1$},\\[3ex]
 \begin{pmatrix}1&0\\
ie^{-\frac{2i}{\e}\phi(\lb;x,t)}&1
\end{pmatrix},&\mbox{ as $\lb\in\Sigma_2$,}\\[3ex]
\begin{pmatrix}e^{-\frac{2i}{\e}\phi_+(\lb;x,t)}&i\\
i&0
\end{pmatrix},&\mbox{ as $\lambda\in (u_c,0)$,}
\end{array}
\end{cases}
\end{equation}
\item[(c)] if we perform the double scaling limit where we let $\epsilon\to 0$ and at the
same time we let $x\to x_c$ and $t\to t_c$ in such a way that
\begin{equation}\label{doublescaling}
\lim\dfrac{x- x_c-6u_c (t-t_c)}{(8k\e^6)^{1/7}}=X, \qquad \lim\dfrac{6(t-t_c)}{(4k^3\epsilon^4)^{1/7}}= T,
\end{equation}
the following matching condition holds,
\begin{equation}\label{RHP P:c}
P(\lambda)P^{(\infty)}(\lambda)^{-1}\to I, \qquad \mbox{ for
$\lambda\in \partial \mathcal U$.}
\end{equation}
\end{itemize}
We will show later on that $v_P$ approximates $v_S$ in the following
sense,
\begin{equation}\label{asymptotics vSP}
v_S(z)v_P^{-1}(z)=I+\bigO(\e), \qquad\mbox{ uniformly for
$z\in\mathcal U\cap \Sigma_S$ as $\e\to 0$}.
\end{equation}

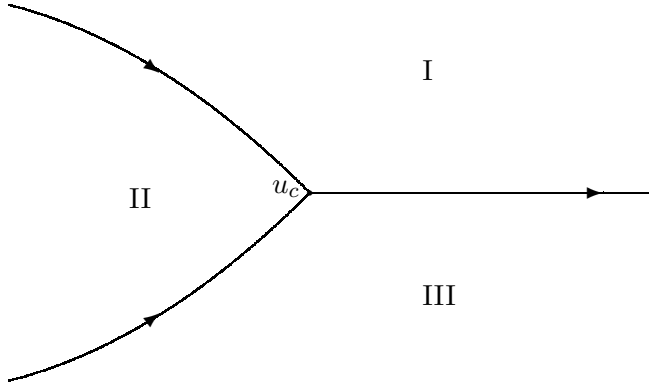
\begin{figure}[t]
\begin{center}
    \setlength{\unitlength}{1mm}
    \begin{picture}(137.5,26)(22,11.5)

        \put(85,25){\thicklines\circle*{.8}} \put(80,25){$u_c$}
        \put(85,25){\line(1,0){45}}
        \put(124,25){\thicklines\vector(1,0){.0001}}
        \qbezier(85,25)(65,45)(45,50) \put(65,41){\thicklines\vector(4,-3){.0001}}
        \qbezier(85,25)(65,5)(45,0) \put(65,9){\thicklines\vector(4,3){.0001}}

\put(100,40){I} \put(61,23){II} \put(100,10){III}
    \end{picture}\vspace{1cm}
    \caption{The jump contour for $P$}
    \label{Figure: contour P}
\end{center}
\end{figure}

\subsubsection{Model RH problem associated with $P_I^2$}

Recall that the $P_I^2$ equation is the following fourth order differential
equation for $U=U(X,T)$,
\begin{equation}\label{PI2}
    X=T\, U-\left[\frac{1}{6}U^3+\frac{1}{24}(U_X^2+2UU_{XX})
        +\frac{1}{240}U_{XXXX}\right].
\end{equation}
It was conjectured in \cite{dubcr} and proven in \cite{CV1} that
this equation has a real pole-free solution $U(X,T)$. In agreement
with the conjecture of Dubrovin, it follows from previous results in
\cite{Menikoff, Moore} that the pole-free solution is unique. The RH
problem characterizing this solution is the following, see
\cite{Kapaev,CV1}.

\subsubsection*{RH problem for $\Psi$:}

\begin{itemize}
    \item[(a)] $\Psi=\Psi(\zeta;x,t)$ is analytic for $\zeta\in\mathbb{C}\setminus\Gamma$, with $\Gamma$ as shown in Figure
    \ref{figure: contour gamma}.
    \item[(b)] $\Psi$ satisfies the following jump relations on
    $\Gamma$,
    \begin{align}
        \label{RHP Psi: b1}
        &\Psi_+(\zeta)=\Psi_-(\zeta)
        \begin{pmatrix}
            0 & 1 \\
            -1 & 0
        \end{pmatrix},& \mbox{for $\zeta\in\Gamma_3$,} \\[1ex]
        \label{RHP Psi: b2}
        &\Psi_+(\zeta)=\Psi_-(\zeta)
        \begin{pmatrix}
            1 & 1 \\
            0 & 1
        \end{pmatrix},& \mbox{for $\zeta\in\Gamma_1$,} \\[1ex]
        \label{RHP Psi: b3}
        &\Psi_+(\zeta)=\Psi_-(\zeta)
        \begin{pmatrix}
            1 & 0 \\
            1 & 1
        \end{pmatrix},& \mbox{for $\zeta\in\Gamma_2\cup \Gamma_4$.}
    \end{align}
    \item[(c)] $\Psi$ has the following behavior at infinity,
    uniformly for $(X,T)$ in compact subsets of $\mathbb C^2\setminus \mathcal
    P$, where $\mathcal P$ denotes the set of poles of $U$,
    \begin{multline}\label{RHP Psi: c}
        \Psi(\zeta)=\zeta^{-\frac{1}{4}\sigma_3}N
        \left(I+Q\sigma_3\zeta^{-1/2}
        +\frac{1}{2}\begin{pmatrix}Q^2 & iU\\-iU &
        Q^2\end{pmatrix}\zeta^{-1}\right. \\
        \left. +R\sigma_3\zeta^{-3/2}+\bigO(\zeta^{2})\right)
        e^{-\theta(\zeta;X,T)\sigma_3},
    \end{multline}
    where $U=U(X,T)$ is the real pole-free solution of the $P_I^2$ equation (\ref{PI2}),
    $\partial_X Q(X,T)=U(X,T)$, $R$
    is some unimportant function of $X$ and $T$, and $N$ and $\theta$ are given by
    \begin{equation}\label{def N theta}
    N=\frac{1}{\sqrt 2}\begin{pmatrix}1&1\\-1&1\end{pmatrix}e^{-\frac{\pi
    i}{4}\sigma_3}, \qquad
    \theta(\zeta;X,T)=\frac{1}{105}\zeta^{7/2}-\frac{T}{3}\zeta^{3/2}+X\zeta^{1/2}.
    \end{equation}
\end{itemize}
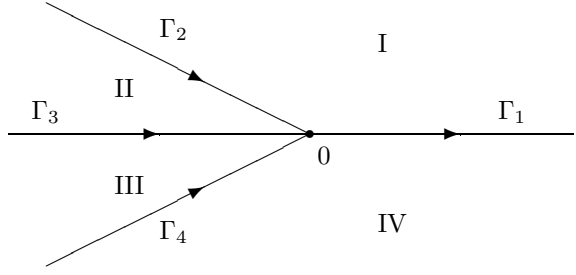
\begin{figure}[t]
    \begin{center}
    \setlength{\unitlength}{1mm}
    \begin{picture}(95,47)(0,2)
        \put(30,38){\small $\Gamma_2$}
        \put(13,27){\small $\Gamma_3$}
        \put(30,11){\small $\Gamma_4$}
        \put(75,27){\small $\Gamma_1$}

        \put(59,36){\small $\I$}
        \put(24,30){\small $\II$}
        \put(24,17){\small $\III$}
        \put(59,12){\small $\IV$}

        \put(50,25){\thicklines\circle*{.9}}
        \put(51,21){\small 0}

        \put(50,25){\line(-2,1){35}} \put(36,32){\thicklines\vector(2,-1){.0001}}
        \put(50,25){\line(-2,-1){35}} \put(36,18){\thicklines\vector(2,1){.0001}}
        \put(50,25){\line(-1,0){40}} \put(30,25){\thicklines\vector(1,0){.0001}}
        \put(50,25){\line(1,0){35}} \put(70,25){\thicklines\vector(1,0){.0001}}
    \end{picture}
    \caption{The jump contour for $\Gamma$ for $\Psi$}
        \label{figure: contour gamma}
    \end{center}
\end{figure}

\begin{remark}\label{remark: poles}
It is important for us that, for each $(X,T)\in\mathbb R^2$, there
is a neighborhood $\mathcal A\subset \mathbb C^2$ of $(X,T)$ such
that $\Psi(.;x,t)$ is well-defined for all $(x,t)\in\mathcal A$ and
such that the asymptotic condition (\ref{RHP Psi: c}) holds
uniformly for $(x,t)\in\mathcal A$. This is only guaranteed by the
fact that $U(X,T)$ has no real poles \cite{CV1}.
\end{remark}
\begin{remark}
There is some freedom in choosing the angles that $\Gamma_2$ and $\Gamma_4$ make with the negative real axis, as long as those angles are in absolute value smaller than $\frac{2\pi}{7}$ (which is the angle for the Stokes lines closest to the negative real line). For simplicity, we fix $\Gamma_2=\{z\in\mathbb C: \arg z=\frac{6\pi}{7}\}$ and $\Gamma_4=\{z\in\mathbb C: \arg z=\frac{-6\pi}{7}\}$ to be the so-called anti-Stokes lines.
\end{remark}

The RH solution $\Psi$ provides in each sector fundamental bases of solutions to the Lax equations associated to the $P_I^2$ equation. Those equations are
\begin{equation}\label{lax pair}\Psi_\zeta=A\Psi, \qquad \Psi_X=B\Psi, \qquad \Psi_T=C\Psi,\end{equation}
where the matrices $A$, $B$, and $C$ are polynomials in $\zeta$, given by
{\small \begin{align*}\label{introduction: U}
   & A = \frac{1}{240}
    \begin{pmatrix}
        -4U_X \zeta-(12UU_X+U_{XXX}) &
        8\zeta^2+8U\zeta+(12U^2+2U_{XX}-120T) \\[1ex]
        A_{21}
        & 4U_X \zeta+(12UU_X+U_{XXX})
    \end{pmatrix},\\[3ex]
    & A_{21} =
    8\zeta^3-8U\zeta^2-(4U^2+2U_{XX}+120T)\zeta+
    (16U^3-2U_X^2+4UU_{XX}+240X),
\end{align*}
}
\begin{equation*}\label{introduction: W}
    B =\begin{pmatrix}
        0 & 1 \\
        \zeta-2U & 0
    \end{pmatrix},\qquad C =\frac{1}{6}\begin{pmatrix}
        U_X & -2\zeta -2U \\
        -2\zeta^2+2U\zeta +c(X,T) & -U_X
    \end{pmatrix}.
\end{equation*}
Compatibility of the first and the second equation in (\ref{lax pair}) shows (after a not so short derivation) that $U$ solves the $P_I^2$ equation, while compatibility of the first and the third equation in (\ref{lax pair}) shows that $U(X,T)$ solves the KdV equation normalized to the form $U_T+UU_X+\frac{1}{12} U_{XXX}=0$, see also \cite{dubcr}. One can make the rescalings as in formula (\ref{univer}) to obtain the KdV equation in the form (\ref{KdV}).

\subsubsection{Modified model RH problem}
In order to obtain a RH problem that we can use to construct the
local parametrix near $u_c$, we transform the RH problem for $\Psi$ by
defining $\Phi$ in the following way,
\begin{align}\label{def Phi}
&\Phi(\zeta;X,T)= \begin{cases} e^{\frac{\pi
i}{4}\sigma_3}\Psi(\zeta)e^{\theta(\zeta;X,T)\sigma_3}
\begin{pmatrix}0&-1\\1&0\end{pmatrix}e^{\frac{\pi
i}{4}\sigma_3},&\mbox{ as $\Im\zeta >0$,}\\[3ex]
e^{-\frac{\pi
i}{4}\sigma_3}\Psi(\zeta)e^{\theta(\zeta;X,T)\sigma_3}e^{\frac{\pi
i}{4}\sigma_3},&\mbox{ as $\Im\zeta <0$.}
\end{cases}
\end{align}
Clearly $\Phi$ is analytic in $\mathbb C\setminus \Gamma$, where
$\Gamma$ is the jump contour for $\Psi$ shown in Figure
\ref{figure: contour gamma}. Using the fact that
$\theta_+(\zeta;X,T)=-\theta_-(\zeta;X,T)$ for $\zeta\in(-\infty,
0)$ and (\ref{RHP Psi: b1}), one can check using the definition
(\ref{def Phi}) that $\Phi$ has no jump on $(-\infty,0)$
and is thus analytic on $(-\infty,0)$. Using (\ref{RHP Psi: b2}),
(\ref{RHP Psi: b3}), and (\ref{def Phi}), one checks that the
jumps along the other rays in the contour do not vanish but are
modified by the transformation. By (\ref{RHP Psi: c}) and
(\ref{def Phi}), also the asymptotic behavior of $\Phi$ is
different from the asymptotic behavior of $\Psi$. $\Phi$ satisfies
the following RH problem.

\subsubsection*{RH problem for $\Phi$}
\begin{itemize}
    \item[(a)] $\Phi$ is analytic for $\zeta\in\mathbb{C}
    \setminus\widehat\Gamma$, with $\widehat\Gamma=\Gamma_1\cup\Gamma_2\cap\Gamma_4$.
    \item[(b)] $\Phi$ satisfies the following jump relations on
    $\widehat\Gamma$,
    \begin{align}
        \label{RHP Phi: b1}
        &\Phi_+(\zeta)=\Phi_-(\zeta)\begin{pmatrix}
            e^{-2\theta(\zeta;X,T)} & i \\
            i & 0
        \end{pmatrix}
        ,& \mbox{for $\zeta\in\Gamma_1$,} \\[1ex]
        \label{RHP Phi: b2}
        &\Phi_+(\zeta)=\Phi_-(\zeta)\begin{pmatrix}
            1 & ie^{2\theta(\zeta;X,T)} \\
            0 & 1
        \end{pmatrix}
        ,& \mbox{for $\zeta\in\Gamma_2$.}
        \\[1ex]
        \label{RHP Phi: b3}
        &\Phi_+(\zeta)=\Phi_-(\zeta)
        \begin{pmatrix}
            1 & 0 \\
            ie^{2\theta(\zeta;X,T)} & 1
        \end{pmatrix},& \mbox{for $\zeta\in\Gamma_4$.}
    \end{align}
    \item[(c)] $\Phi$ has the following behavior at infinity,
    \begin{multline}\label{RHP Phi: c1}
        \Phi(\zeta)=e^{-\frac{\pi i}{4}\sigma_3}\zeta^{-\frac{1}{4}\sigma_3}N\\
        \qquad\times\quad
        \left(I+Q\sigma_3\zeta^{-1/2}
        +\frac{1}{2}\begin{pmatrix}Q^2 & iU\\-iU &
        Q^2\end{pmatrix}\zeta^{-1}+R\sigma_3\zeta^{-3/2}+\bigO(\zeta^{2})\right) \\
        \qquad\qquad\qquad\times
        \begin{cases}\begin{pmatrix}0&-1\\1&0\end{pmatrix}e^{\frac{\pi i}{4}\sigma_3},
        &\mbox{as $\zeta\to\infty$ with $\Im\zeta>0$,}\\[3ex]
        e^{\frac{\pi
        i}{4}\sigma_3},&\mbox{as $\zeta\to\infty$ with $\Im\zeta<0$,}
        \end{cases}
    \end{multline}
    which can be rewritten in the following way, both for
    $\zeta\to\infty$ in the upper and the lower half plane,
    \begin{multline}\label{RHP Phi: c}
        \Phi(\zeta)=\frac{1}{\sqrt
        2}(-\zeta)^{-\frac{1}{4}\sigma_3}\begin{pmatrix}1&1\\-1&1\end{pmatrix}\left(I+iQ\sigma_3(-\zeta)^{-1/2}-\frac{1}{2}
        \begin{pmatrix}Q^2&U\\U&Q^2\end{pmatrix}(-\zeta)^{-1}\right. \\
        \left. +iR\sigma_3(-\zeta)^{-3/2}+\bigO(\zeta^{-2})\right).
    \end{multline}
    This behavior holds uniformly for $(X,T)$ in compact subsets of
    $\mathbb C^2\setminus \mathcal P$, where $\mathcal P$ is the set
    of poles of $U$.
\end{itemize}

\subsubsection{Construction of the parametrix}

We are now ready to specify the form of the parametrix $P$. Let $P$ be of the following form,
\begin{equation}\label{definition hatP}
P(\lambda)=E(\lambda;\epsilon)\Phi(\epsilon^{-2/7}f(\lambda);
\epsilon^{-6/7}g_1(\lambda;\tilde{x}),\epsilon^{-4/7}g_2(\lambda;t)),
\end{equation}
where $E$, $f$, $g_1$, and $g_2$ are analytic in $\mathcal U$ and $\tilde{x}=x-6u_ct$.
Furthermore we require that $f(\mathbb R\cap \mathcal U)\subset
\mathbb R$ with $f(u_c)=0$ and $f'(u_c)>0$, in other words $f$ is
a conformal mapping from $\mathcal U$ to a neighborhood of the
origin. We will determine the precise form of $E$, $f$, $g_1$, and
$g_2$ later.

\medskip

First we specify our choice of the contour $\Sigma_S\cap \mathcal
U$ by requiring that $f(\Sigma_S\cap \mathcal U)\subset
\tilde\Gamma$, which we can do since $f$ is a conformal mapping.
By this construction $P$ and $S$ have their jumps on the
same contour in $\mathcal U$. We will now use the freedom we still
have in defining $f$, $g_1$, and $g_2$ in order to create jumps
for $P$ that are the same as the ones specified in (\ref{RHP P:b})-(\ref{vP}). Afterwards we will define $E$ in such a way that the
matching condition (\ref{RHP P:c}) is satisfied as well.

\subsubsection{Definition of $f$, $g_1$, and $g_2$}

Our goal is to define $f$, $g_1$, and $g_2$ in such a way that
\begin{equation}\label{condition f g 1}
\begin{array}{ll}
\theta(\e^{-2/7}f(\lambda);\e^{-6/7}g_1(\lambda;\tilde{x}),\e^{-4/7}g_2(\lambda;t))=\frac{i}{\e}
\phi_+(\lambda;\tilde{x},t),&\mbox{for $\lambda\in (u_c, 0)\cap\mathcal
U$,}\\[2ex]
\theta(\e^{-2/7}f(\lambda);\e^{-6/7}g_1(\lambda;\tilde{x}),\e^{-4/7}g_2(\lambda;t))=\frac{i}{\e}
\phi(\lambda;\tilde{x},t),&\mbox{for $\lambda\in\Sigma_1\cap\mathcal
U$,}\\[2ex]
\theta(\e^{-2/7}f(\lambda);\e^{-6/7}g_1(\lambda;\tilde{x}),\e^{-4/7}g_2(\lambda;t))=-\frac{i}{\e}
\phi(\lambda;\tilde{x},t),&\mbox{for $\lambda\in\Sigma_2\cap\mathcal U$,}\\
\end{array}
\end{equation}
with $\phi$ and $\theta$ defined by (\ref{definition phi}) and
(\ref{def N theta}). For simplicity, we changed variables for $\phi$ and consider it now as a function of $\tilde x$ instead of $x$, with
\[
\tilde{x}=x-6u_ct.\] We should now compare the jump matrices
(\ref{vP}) for $P$ to the 'model' jump matrices (\ref{RHP Phi:
b1})-(\ref{RHP Phi: b3}) for $\Phi$. The above conditions
(\ref{condition f g 1}) imply that the parametrix $P$, defined as in
(\ref{definition hatP}), satisfies the requested jump condition
(\ref{RHP P:b}).

\medskip

With $\theta(\zeta)$ having its branch cut along the negative real
line, one verifies that the equations in (\ref{condition f g 1})
are satisfied if
\begin{equation}\label{condition f g 2}
\theta(-f(\lambda);-g_1(\lambda;\tilde{x}),g_2(\lambda;t))=-\phi(\lambda;\tilde{x},t),\qquad \mbox{for $\lambda\in\mathcal
U$,}
\end{equation}
In agreement with this condition, we first define $f$ in
such a way that
\begin{equation}\label{def f}
\frac{1}{105}\left(-f(\lambda)\right)^{7/2}=-\phi(\lambda;\tilde{x}_c,t_c).
\end{equation}
Using (\ref{uxx1}) and (\ref{phi}), integrating by parts twice leads to
\begin{multline}
\label{phi1}
\phi(\lambda;\tilde{x},t)=\sqrt{u_c-\lambda}(\tilde{x}-\tilde{x}_c)+4(u_c-\lambda)^{\frac{3}{2}}
(t-t_c)\\
+\dfrac{4}{15}\int_{\lambda}^{u_c}f'''_-(\xi)(\xi-\lambda)^{\frac{5}{2}}d\xi.
\end{multline}
Keeping track of the branch cuts for the fractional powers, this implies that $f$ is analytic
in $\mathcal U$ with
 \begin{equation}\label{fu} f(u_c)=0, \qquad
f'(u_c)=(8k)^{\frac{2}{7}},\qquad
k=-f_-'''(u_c)>0.
\end{equation}
Now we can define $g_2$ by requiring that
\begin{equation}\label{def g2}
\frac{1}{3}g_2(\lambda;t)(-f(\lambda))^{3/2}=\phi(\lambda;\tilde{x}_c,t)-
\phi(\lambda;\tilde{x}_c,t_c)=4(u_c-\lambda)^{\frac{3}{2}}(t-t_c).
\end{equation}
Since $f$ is a conformal mapping, this defines $g_2$ analytically
in $\mathcal U$, with
\begin{equation}\label{g2u}g_2(u_c;t)=\dfrac{6(t-t_c)}{(4k^3)^{1/7}}.
\end{equation}
Finally we define
$g_1$ by the equation
\begin{equation}\label{def g1}
g_1(\lambda;\tilde{x},t)(-f(\lambda))^{1/2}=\phi(\lambda;\tilde{x},t)-
\phi(\lambda;\tilde{x}_c,t)=\sqrt{u_c-\lambda}(\tilde{x}-\tilde{x}_c).
\end{equation}
Again using the fact that $f$ is a conformal mapping, $g_1$ is
analytic in $\mathcal U$, with
\begin{equation}\label{g1u}g_1(u_c;\tilde{x},t)=\dfrac{x-x_c-6u_c(t-t_c)}{(8k)^{\frac{1}{7}}}.\end{equation}
Summing up (\ref{def f}), (\ref{def g2}), and (\ref{def g1}), we
find using (\ref{def N theta}) that indeed condition
(\ref{condition f g 2}) is satisfied. This means that the jump
conditions (\ref{RHP P:b})-(\ref{vP}) for $P$ are valid.

\medskip

The function $\Phi(\zeta;X,T)$ is not defined for values of $X$ and $T$ where $U(X,T)$ has a pole.
In order to ensure that the parametrix is
well-defined, we need to know that $U$ has no pole at $(\epsilon^{-6/7}g_1(\lambda;\tilde{x},t),\epsilon^{-4/7}g_2(\lambda;t))$.
In the double scaling limit where $\epsilon\to 0$ and at the
same time $x\to x_c$ and $t\to t_c$ in such a way that
\[
\lim\dfrac{x- x_c-6u_c (t-t_c)}{(8k\e^6)^{1/7}}=X, \qquad \lim\dfrac{6(t-t_c)}{(4k^3\epsilon^4)^{1/7}}= T,
\]
it follows from (\ref{def g2})-(\ref{g1u}) that
\[\epsilon^{-6/7}g_1(\lambda;\tilde{x},t)\to X, \qquad \epsilon^{-4/7}g_2(\lambda;t)\to T,\qquad \mbox{ as $\lambda\to u_c$.}\]
From the fact that $U(X,T)$ is meromorphic
both in $X$ and $T$ and that it has no poles for real values of
$X,T$, it follows that there is a pole-free neighborhood of $(X,T)$ in $\mathbb C^2$ (cf.\ Remark \ref{remark: poles}), in which
$(\epsilon^{-6/7}g_1(\lambda;\tilde{x},t),\epsilon^{-4/7}g_2(\lambda;t))$ is contained for sufficiently small $\epsilon$, provided that $\lambda$ lies in a sufficiently small neighborhood $\mathcal U$ of $u_c$.

\medskip

In order to satisfy also the matching condition (\ref{RHP P:c}), we
define the analytic pre-factor $E$ by
\begin{equation}
E(\lambda;\epsilon)=\frac{1}{\sqrt
2}P^{\infty}(\lambda)\begin{pmatrix}1&-1\\1&1\end{pmatrix}
(-\epsilon^{-2/7}f(\lambda))^{\frac{\sigma_3}{4}}.
\end{equation}
Using the definition (\ref{def Pinfty}) of $P^{(\infty)}$, one
checks directly that $E$ is analytic in $\mathcal U$, as it should be in order to have a suitable parametrix. The
matching condition (\ref{RHP P:c}) for $\lambda\in\partial \mathcal U$ in the
double scaling limit can now be specified as follows,
\begin{multline}\label{RHP P:c2}
P(\lambda)P^{(\infty)}(\lambda)^{-1}=P^{\infty}(\lambda)\left(I+iQ\sigma_3(-f(\lambda))^{-1/2}\epsilon^{1/7}\right.\\
\left.-\frac{1}{2}
        \begin{pmatrix}Q^2&U\\U&Q^2\end{pmatrix}(-f(\lambda))^{-1}\epsilon^{2/7}
        +iR\sigma_3(-f(\lambda))^{-3/2}\epsilon^{3/7}
        +\bigO(\epsilon^{4/7})\right)
        P^{\infty}(\lambda)^{-1}.
\end{multline}
Here we used the abbreviations
\begin{align}\label{abbrevyqr}
&U=U(\epsilon^{-6/7}g_1(\lb;\tilde{x},t),\epsilon^{-4/7}g_2(\lb;t)),
\\ &Q=Q(\epsilon^{-6/7}g_1(\lb;\tilde{x},t),\epsilon^{-4/7}g_2(\lb;t)),
\\
&R=R(\epsilon^{-6/7}g_1(\lb;\tilde{x},t),\epsilon^{-4/7}g_2(\lb;t)).
\end{align}
This ends the construction of the local parametrix.

\subsection{Final transformation}

We define
\begin{align}\label{def R}&
R(\lambda)=\begin{cases}S(\lambda)P^{(\infty)}(\lambda)^{-1},
& \mbox{ as $\lambda\in \mathbb C\setminus \mathcal U$,}\\
S(\lambda)P(\lambda)^{-1}, & \mbox{ as $\lambda\in \mathcal U$.}
\end{cases}
\end{align}
Note first that outside parametrix has been constructed in such a way that $R$ has no jump on
$(0,+\infty)$.
For $z\in\mathcal U\cap\Sigma_S$, we have that
\begin{equation}R_-^{-1}(z)R_+(z)=P_-(z)v_S(z)v_P^{-1}(z)P_-^{-1}(z).\end{equation}
On one hand it follows from the construction of the parametrix that
$P_-(z)$ is uniformly bounded for $z\in\mathcal U\cap\Sigma_S$. On
the other hand
\begin{equation}
v_Sv_P^{-1}=\begin{cases}
\begin{array}{ll}
\begin{pmatrix}1&i(\kappa -1)e^{\frac{2i}{\e}\phi}\\
0&1
\end{pmatrix},&\mbox{ on $\Sigma_1\cap \mathcal U$},\\[3ex]
 \begin{pmatrix}1&0\\
i(\kappa^* -1)e^{-\frac{2i}{\e}\phi}&1
\end{pmatrix},&\mbox{ on $\Sigma_2\cap \mathcal U$,}\\[3ex]
\begin{pmatrix}\kappa&i(\kappa-1)e^{-\frac{2i}{\e}\phi_+}\\
-i(1-|r|^2)e^{\frac{2i}{\e}\phi_+}&\kappa^*+(1-|r|^2)
\end{pmatrix},&\mbox{ on $(u_c,0)\cap \mathcal U$.}
\end{array}
\end{cases}
\end{equation}
Except for the $21$-entry on $(u_c,0)$, the exponentials in the above matrices are uniformly bounded on the
jump contours inside $\mathcal U$ because of (\ref{condition f g 1})
and the fact that the jumps for $\Phi$ are uniformly bounded on the
jump contour $\widehat\Gamma$. Here the WKB
approximation (\ref{kappa1}) for $\kappa$ ensures that
$v_Sv_P^{-1}=I+\bigO(\e)$. If in addition we use (\ref{WKB r}) and Proposition \ref{prop phi}, the
same follows for the $21$-entry on $(u_c,0)\cap \mathcal U$. By Proposition \ref{prop
vs}, one checks that the RH problem for $R$ has the following form.

\subsubsection*{RH problem for $R$}
\begin{itemize}
\item[(a)]$R$ is analytic in $\mathbb C\setminus\Sigma_R$, with
$\Sigma_R=(\Sigma_S\cup\partial\mathcal U)\setminus (0,+\infty)$ as shown in Figure \ref{figure: sigmaR},
\item[(b)]$R_+(\lambda)=R_-(\lambda)v_R(\lambda)$, where the jump matrix $v_R$ has the following double scaling asymptotics,
\begin{align}\label{vR}
&
v_R(\lambda)=\begin{cases}P(\lambda)P^{(\infty)}(\lambda)^{-1},&\mbox{
as $\lambda\in \mathbb \partial \mathcal U$,}\\
I+\bigO(\epsilon),&\mbox{ as $\lambda\in
\Sigma_R\setminus \partial \mathcal U$.}
\end{cases}
\end{align}
\item[(c)]$R(\lambda)\to I$ as $\lambda\to\infty$.
\end{itemize}

This is the final RH problem in our analysis. In the next section, we will obtain uniform asymptotics for $R$ which will enable us to prove Theorem \ref{theorem: main}.

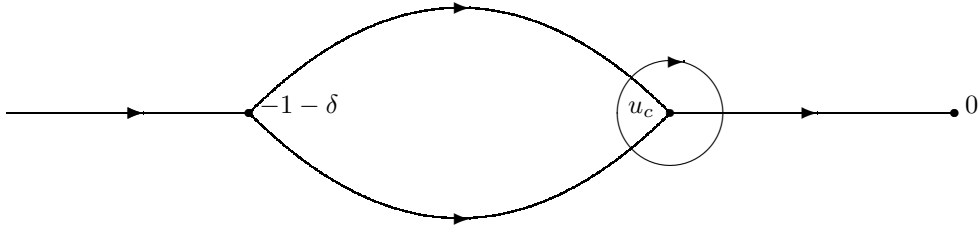
\begin{figure}[t]
\begin{center}

    \setlength{\unitlength}{1.4mm}
    \begin{picture}(137.5,26)(22,11.5)
        \put(112,25){\thicklines\circle*{.8}}
        \put(45,25){\thicklines\circle*{.8}}
        \put(22,25){\line(1,0){23}}
        \put(35,25){\thicklines\vector(1,0){.0001}}
        \put(85,25){\circle{15}}
\put(86.3,29.9){\thicklines\vector(1,0){.0001}}
        \put(112,25){\thicklines\circle*{.8}}
        \put(113,25){\small $0$}
        \put(45,25){\thicklines\circle*{.8}}
        \put(46,25){\small $-1-\delta$}
        \put(85,25){\thicklines\circle*{.8}} \put(81,25){\small $u_c$}
        \put(99,25){\thicklines\vector(1,0){.0001}}
        \put(85,25){\line(1,0){27}}
        \put(22,25){\line(1,0){23}}
        \put(35,25){\thicklines\vector(1,0){.0001}}
        \qbezier(45,25)(65,45)(85,25) \put(66,35){\thicklines\vector(1,0){.0001}}
        \qbezier(45,25)(65,5)(85,25) \put(66,15){\thicklines\vector(1,0){.0001}}
    \end{picture}
    \caption{The contour $\Sigma_R$ after the third and final
        transformation.}
    \label{figure: sigmaR}
\end{center}
\end{figure}

\section{Proof of Theorem \ref{theorem: main}}\label{section: proof}

Using (\ref{RHP P:c2}), we can expand the jump matrix $v_R$ in
fractional powers of $\e$ in the double scaling limit,
\begin{equation}\label{vRexpansion}
v_R(\lambda)=I+\epsilon^{1/7}\Delta^{(1)}(\lambda)+\epsilon^{2/7}\Delta^{(2)}(\lambda)+
\bigO(\epsilon^{3/7}),
\end{equation}
with
\begin{align}\label{delta1}
&\Delta^{(1)}(\lambda)=iQ \cdot
(-f(\lambda))^{-1/2}P^{(\infty)}(\lambda)\sigma_3
P^{(\infty)}(\lambda)^{-1},\\
\label{delta2}&\Delta^{(2)}(\lambda)=-\frac{1}{2}(-f(\lambda))^{-1}P^{(\infty)}(\lambda)
\begin{pmatrix}Q^2&Y\\Y&Q^2\end{pmatrix}
P^{(\infty)}(\lambda)^{-1},
\end{align}
for $\lambda\in\partial \mathcal U$, and
\begin{equation}
\Delta^{(1)}(\lambda)=\Delta^{(2)}(\lambda)=0, \qquad \mbox{for $\lambda\in\Sigma_R\setminus\partial\mathcal U$,}
\end{equation}
 since the jump matrices are equal to $I$ up to an error of $\bigO(\e)$ on the other parts of the contour.
Note that the functions $\Delta^{(1)}$ and $\Delta^{(2)}$ are meromorphic functions in $\mathcal U$ with poles at $u_c$.
\medskip

It is a well-known result that the uniform asymptotic expansion (\ref{vRexpansion}) for the jump matrix yields an asymptotic expansion of the same form for the RH solution $R$, uniformly for $\lambda\in\mathbb C\setminus \Sigma_R$,
\begin{equation}\label{Rexpansionepsilon}
R(\lambda)=I+\epsilon^{1/7}R^{(1)}(\lambda)+\epsilon^{2/7}R^{(2)}(\lambda)
+\bigO(\epsilon^{3/7}).
\end{equation}
This can proven exactly as in \cite{Deift, DKMVZ1} by estimating Cauchy-type operators associated to the RH problem.

\begin{remark}
We should note that, in \cite{Deift, DKMVZ2, DKMVZ1}, the local
parametrices erase the jumps inside disks near the special points
completely. Because of the $\bigO(\e)$-error of the reflection
coefficient to its WKB approximation, we have not been able to
remove the jumps in $\mathcal U$ completely. This leads to a final
RH problem which is similar to the one in \cite{BDJ}. However the
remaining jumps in the interior of $\mathcal U$ are of smaller order
than the relevant jumps on $\partial\mathcal U$, and therefore this does not
cause any further problems.
\end{remark}

The compatibility of the expansion (\ref{vRexpansion}) for the jump
matrix with the expansion (\ref{Rexpansionepsilon}) for the RH
solution leads, using the jump relation
$R_+(\lambda)=R_-(\lambda)v_R(\lambda)$, to the following conditions
for $\lambda\in\partial \mathcal U$,
\begin{align}\label{RHPR1}
& R_+^{(1)}(\lambda)=R_-^{(1)}(\lambda) + \Delta^{(1)}(\lambda),\\
\label{RHPR2}& R_+^{(2)}(\lambda)=R_-^{(2)}(\lambda) +
R_-^{(1)}(\lambda)\Delta^{(1)}(\lambda)+ \Delta^{(2)}(\lambda).
\end{align}
Note in addition that $R(\lambda)\to I$ as $\lambda\to\infty$, and
thus also $R^{(j)}(\lambda)\to I$ for $j=1,2$. The jump and asymptotic conditions
constitute additive RH problems for $R^{(1)}$ and $R^{(2)}$. By inspection we
see that their unique solutions are given by
\begin{align}\label{R1}
& R^{(1)}(\lambda)= \begin{cases}\frac{1}{\lambda
-u_c}\Res(\Delta^{(1)};u_c),
&\mbox{ as $\lambda\in\mathbb C\setminus \mathcal U$}\\
\frac{1}{\lambda -u_c}\Res(\Delta^{(1)};u_c)-\Delta^{(1)}(\lambda),
&\mbox{ as $\lambda\in \mathcal U$,}\end{cases}\\
\label{R2} & R^{(2)}(\lambda)= \nonumber\\&\begin{cases} \frac{1}{\lambda
-u_c}\Res(R^{(1)}\Delta^{(1)}+\Delta^{(2)};u_c),
&\mbox{ as $\lambda\in\mathbb C\setminus \mathcal U$,}\\
\frac{1}{\lambda
-u_c}\Res(R^{(1)}\Delta^{(1)}+\Delta^{(2)};u_c)-R^{(1)}\Delta^{(1)}(\lambda)
-\Delta^{(2)}(\lambda),
&\mbox{ as $\lambda\in \mathcal U$.}\\
\end{cases}
\end{align}
After a straightforward calculation we find using (\ref{delta1}),
(\ref{delta2}), and (\ref{def Pinfty}) that, for $\lambda\in\mathbb
C\setminus \mathcal U$,
\begin{align}
& R^{(1)}(\lambda)=-Q f'(u_c)^{-1/2}\frac{1}{\lambda -
u_c}\begin{pmatrix}0 & 1\\0&0
\end{pmatrix},\label{eqR1}\\
& R^{(2)}(\lambda)=\begin{pmatrix}* & 0\\ * & *
\end{pmatrix},\label{eqR2}
\end{align}
where the $*$'s denote unimportant entries, and where we have now written
\[Q=Q(\epsilon^{-6/7}g_1(u_c;\tilde{x},t),\epsilon^{-4/7}g_2(u_c;t)).\]
We observe already that the main sub-leading terms in the asymptotic
expansion for $R$ are determined completely by the matching of the
local parametrix $P$ with the outside parametrix $P^{(\infty)}$.

\medskip

Besides the expansion in negative powers of $\e$, $R$ admits also an
expansion in negative powers of $\lb$ as $\lambda\to\infty$
\cite{Deift},
\begin{equation}\label{Rexpansion}
R(\lambda)=I+\frac{R_1}{\lambda} +\bigO(\lambda^{-2}), \qquad \mbox{
as $\lambda\to\infty$}.
\end{equation}
Compatibility of the small $\epsilon$-expansion (\ref{Rexpansionepsilon}) with the large $\lambda$-expansion
 (\ref{Rexpansion})
learns us the following in the double scaling limit,
\begin{equation}\label{R1expansion}
R_{1,12}(x,t,\epsilon)=-\epsilon^{1/7}Q\,
f'(u_c)^{-1/2}+\bigO(\epsilon^{3/7}).\end{equation} Now it turns out
that the KdV solution $u(x,t,\e)$ is contained in $R_1$. Indeed by
(\ref{def R}), we have that
$S(\lambda)=R(\lambda)P^{(\infty)}(\lambda)$ for large $\lambda$,
which means by (\ref{RHP Pinfty c2}) and (\ref{Rexpansion}) that
\begin{equation}
S_{11}(\lambda)=1+i\frac{R_{1,12}}{(-\lambda)^{1/2}}+\bigO(\lambda^{-1}),
\qquad\mbox{ as $\lambda\to\infty$}.
\end{equation}
In view of (\ref{uS}) we obtain using $\partial_X Q=U$ that
\begin{eqnarray}
\label{uR} u(x,t,\epsilon)&=&u_c-2i\epsilon\partial_x
R_{1,12}(\tilde{x},t,\epsilon)\nonumber \\
&=&u_c+2f'(u_c)^{-1/2}\epsilon^{8/7}\partial_xQ(\epsilon^{-6/7}g_1(u_c;\tilde{x},t),\epsilon^{-4/7}g_2(u_c;t))
+\bigO(\epsilon^{4/7})\nonumber \\
&=&u_c+\left(\dfrac{2\epsilon^2}{k^2}\right)^{1/7}
U(\epsilon^{-6/7}g_1(u_c;\tilde{x},t),\epsilon^{-4/7}g_2(u_c;t))+\bigO(\epsilon^{4/7}).\nonumber
\end{eqnarray}
By (\ref{g2u}) and (\ref{g1u}), this completes the proof of Theorem \ref{theorem: main}.

\section*{Acknowledgements}
We thank B.~Dubrovin and G. Panati for helpful
discussions and hints.  TC is a Postdoctoral Fellow of the Fund for Scientific Research - Flanders (Belgium), and was also supported by FWO project G.0455.04,
K.U. Leuven research grant OT/04/21, and
Belgian Interuniversity Attraction Pole P06/02. TC is also grateful to SISSA for hospitality, and to ESF program MISGAM for supporting his stay at SISSA.
 TG  acknowledges support by the MISGAM program
of the European Science Foundation,  by the RTN ENIGMA and  Italian
COFIN 2004 ``Geometric methods in the theory of nonlinear waves and their applications''.

\end{document}